\documentclass[11pt,a4paper]{article} 
\usepackage{jheppub}
\usepackage[T1]{fontenc}
\usepackage[utf8x]{inputenc}
\usepackage{tikz-cd}
\usepackage{amssymb}
\usepackage{amsfonts}
\usepackage{amsmath}
\usepackage{centernot}
\usepackage{adjustbox}
\usepackage{amssymb,amscd}
\usepackage{mathrsfs}
\usepackage{comment}
\usepackage{amsmath,amsthm}
\usepackage{marginnote}
\usepackage{mathtools}
\usepackage{appendix}
\usepackage{multirow}
\usepackage{bigstrut}

\newtheorem{definition}{Definition}[section]
\newtheorem{theorem}{Theorem}[section]

\newtheorem*{conjecture}{Conjecture}
\newtheorem{proposition}{Proposition}[section]

\newtheorem{lemma}[theorem]{Lemma}

\usepackage{diagbox}

\def\be{\begin{equation}}
\def\ee{\end{equation}}

\def\Tr{{\rm Tr}}

\def\IC{\mathbb{C}}

\def\IN{\mathbb{N}}

\def\IQ{\mathbb{Q}}
\def\IR{{\mathbb{R}}}

\def\IZ{{\mathbb{Z}}}
\def\T{{\mathsf{T}}}

\def\CA{{\cal A}}
\def\CB{{\cal B}}
\def\CC {{\cal C}}

\def\CD {{\cal D}}
\def\CF {{\cal F}}

\def\CW {{\cal W}}

\def\CZ {{\cal Z}}

\def\CH {{\cal H}}

\def\CB {{\cal B}}

\def\CS {{\cal S}}

\def\CZ{{\cal Z}}
\def\la{\langle}
\def\ra{\rangle}

\def\one{{\hbox{ 1\kern-.8mm l}}}

\def\be{\bar{e}}

\def\be{ \begin{equation} }
\def\ee{ \end{equation}}

\def\sD{\mathscr{D}}

\def\l{\Lambda}

\title{When Does A Three-Dimensional Chern-Simons-Witten Theory Have A Time Reversal Symmetry? }
\author{Roman Geiko and Gregory W.~Moore}
                                           \affiliation{Department of Physics and Astronomy, Rutgers University ,\\
                                           126 Frelinghuysen Rd., Piscataway NJ 08855, USA}
\emailAdd{roman.geiko@physics.rutgers.edu} \emailAdd{gmoore@physics.rutgers.edu}

\abstract{In this paper, we completely characterize time-reversal invariant three-dimensional Chern-Simons gauge theories with torus gauge group. At the level of the Lagrangian, toral Chern-Simons theory is defined by an integral lattice, while at the quantum level, it is entirely determined by a quadratic function on a finite Abelian group and an integer mod $24$.  We find that quantum time-reversally symmetric theories can be defined by classical Lagrangians defined by integral lattices which have  self-perpendicular embeddings into a unimodular lattice. 
We find that the quantum toral Chern-Simons theory admits a time-reversal symmetry iff the higher Gauss sums of the associated modular tensor category are real. We conjecture that the reality of the higher Gauss sums is necessary and sufficient for a general non-Abelian Chern-Simons to admit quantum T-symmetry. 
\newline 
}
\begin{document}
\maketitle



\section{Introduction}

Three-dimensional Chern-Simons theory is one of the most interesting and well-studied topological quantum field theories. It has a wide range of applications, from the mathematics of invariants of three-manifolds and knots to the physics of the fractional quantum Hall effect and topological quantum computation. In a series of papers beginning with \cite{Seiberg2016} it was observed that some Chern-Simons theories admit nontrivial time-reversal symmetries\footnote{Topological liquids with a transparent time-reversal symmetry were discussed in \cite{Neupert2011, Neupert2014} and their Chern-Simons description in \cite{Neupert2011S, Huan2016}.}. 
See \cite{Lee:2018eqa, Delmastro:2019vnj} for summaries of known results. In the present paper we completely classify time-reversal invariant (spin) Chern-Simons theories with torus gauge group. 

The basic idea of the classification is very simple (it is reviewed with some more care in the next section). It was demonstrated in \cite{Belov:2005ze} that the 2-3 quantum theory is completely determined by a choice of finite Abelian group $\sD$ (which can be viewed as the group of anyons, or as the group of so-called ``one-form symmetries''), together with a quadratic function on it 
$\hat q: \sD \to \mathbb{Q}/\mathbb{Z}$. A quadratic function is one for which 
\be\label{eq:QuadraticFunction} 
b(x,y):= \hat q(x+y) - \hat q(x) - \hat q(y) + \hat q(0) 
\ee
is bilinear form on $\sD$. (We will assume that the bilinear form is always non-degenerate.) 
The pair $({\sD}, \hat{q} )$ is subject to a suitable equivalence relation recalled in the next section. 
The result of \cite{Belov:2005ze} was extended to the fully extended 0-1-2-3 theory in \cite{Freed:2009qp}. 
The time-reversal invariant theories are simply those such that $({\sD}, \hat q )$ is 
equivalent to $({\sD},-\hat{q})$. We classify such pairs of time-reversally-symmetric group-with-quadratic-function by building on the work of  Wall \cite{Wall1963QuadraticFO}, Miranda \cite{Miranda1984,Miranda2009},  and Kawauchi and Kojima \cite{Kawauchi1980}.
The main result is that a pair is time-reversal if and only if all the higher Gauss sums associated to the quadratic function $\hat{q}$ are real. In the language of (super-) modular tensor categories, we find all ``pointed'' (super-) MTCs that are equivalent to their time-reversal.

It is of course interesting to extend the result to Chern-Simons theories with non-Abelian gauge group. 
We propose that the proper criterion is that the associated modular tensor category is equivalent to 
its reverse, where the braiding transformation has been time-reversed. (This criterion has also been 
proposed in  \cite{Lee:2018eqa,Kong:2022cpy}.) Phrasing the problem this way suggests that a simple criterion is that the modular tensor category should have real higher Gauss sums defined in \cite{NG2019}. We conjecture that any MTC with all higher Gauss sums real correspond to a time-reversal Chern-Simons theory. These higher Gauss sums have played a role in some recent developments in TQFT  \cite{Freed2021,kaidi2021higher}.

This paper is organized as follows. In Section \ref{Conventions}, we set the notation, introduce key definitions and formulate the main problem. In Section \ref{FinForms}, we review the classification of finite forms $(\sD,b)$,
and we describe the forms isomorphic to their negation. We then review the Witt group of finite bilinear forms. In Section \ref{QuantumData}, we assemble time-reversal quantum data from finite forms and discuss the groups of anti-unitary symmetries they enjoy. In Section \ref{Lattices}, we describe Lagrangians whose quantum data is time-reversal symmetric and partially describe classically symmetric theories. In Section \ref{MTC:section}, we review modular tensor categories and formulate the conjecture about the time reversal in non-Abelian Chern-Simons. We summarize our main results in Section \ref{Discussion} and list some open problems. We review some auxiliary maps between finite forms in Appendix \ref{RhoMaps}, Seifert invariants in Appendix \ref{SeifertInv}, signature invariants in Appendix \ref{SignatureInv}, and review the categorical Witt group in Appendix \ref{WittAppendix}.

This paper is dedicated to the memory of Krzysztof Gawedzki. One of us (GM) had the privilege of enjoying many excellent scientific interactions with Krzysztof. He will be sorely missed. 

\section{Toral Chern-Simons-Witten Theory and Finite Forms}
\label{Conventions}
In this letter, we analyze how the toral Chern-Simons theories behave under the time-reversal transformation and study the condition for this family of topological quantum field theories to be time-reversal invariant. The definition of the $U(1)^{N}$ Chern-Simons gauge theory requires a choice of a rank $N$ integral symmetric bilinear form $K$ such that the CS functional informally reads
\begin{equation}
\label{3dCS}
CS_{\l,K}=\frac{1}{4\pi}\int_{X}K_{IJ}A_I\wedge  dA_J\;,
\end{equation}
where $K_{IJ}=K(e_I,e_J)$ is the matrix of $K$ in the given basis for a lattice $\Lambda$. The gauge fields are normalized so that $F/2\pi$ has integral periods.  This theory requires a choice of spin-structure for $X$ if $K$ is odd and does not require this extra choice if $K$ is even. In the former case, the theory is called spin, while in the latter, it is called bosonic. For a proper definition of the functional \eqref{3dCS}, and related matters, see \cite{Belov:2005ze,stirling2008abelian}. Ultimately, $K$ is an element of $H^4(BT, \IZ)$ where $T$ is the Abelian torus. Due to the isomorphism between $H^4(BT, \IZ)$ and the space of integral symmetric bilinear forms, will mean by $K$ the corresponding object in the space forms.    

 To begin, we set some notation. By a lattice $(\l,K)$, we mean a finitely generated free $\IZ$-module $\Lambda$ together with the bilinear form $K$ on it. Here and everywhere below, the form $K$ is assumed to be integral, symmetric, generically odd, and non-degenerate. If it does not lead to confusion, the lattice is denoted simply by $K$. As a gauge theory, $CS_{\l,K}$ does not depend on the choice of basis for $\l$, it rather depends on the isomorphism class of $(\l,K)$ in the space of integral lattices. Classification of integral lattices is a notoriously complicated task, however, as we shall review, the full quantum theory defined by $(\l,K)$ depends on a set of invariants that is much coarser than the isomorphism class. We proceed to a review of the classification of quantum toral Chern-Simons theories and associated lattice invariants first discussed in \cite{Belov:2005ze}.

 As a free module, $\l$ can be embedded in the dual module $\l^{*}=\mbox{Hom}(\Lambda,\IZ)$ by using the non-degenerate form $K$: in terms of the original lattice, the dual lattice reads $\l^*=\{x\in \l\otimes \IQ \,|\, K(x,\mu)\in \IZ\,\; \mbox{for all}\;\mu\in \l\}$. Further, the form $K$ can be extended to the whole $\l^{*}$ by $\IQ$-linearity, giving rise to a form on the dual module:
\begin{align}
    K^{*}:\;\l^{*}\times \l^{*}\to \IQ\;.
\end{align}
 Then, we define the discriminant group $\sD=\l^{*}/\l$, it is a finite Abelian group of order $|\sD|=|\det K|$ which can be equipped with the following torsion form:
\begin{align}
\label{bfinite}
b: \sD\times \sD&\to \IQ/\IZ\,,\\
b(x+\Lambda,y+\Lambda)&=K^{*}(x,y) \;\mbox{mod}\;1\,.
\end{align}
We call the pair consisting of a symmetric non-degenerate bilinear form $b$ and a finite Abelian group $\sD$ simply by a finite bilinear form: the discriminant form on $\sD$ is a finite bilinear form as long as $K$ is non-degenerate.

Given a lattice $(\l,K)$, a quadratic refinement for it is a function $\hat Q: \l \to \IZ$ satisfying 
\begin{align}
\hat Q(\mu +\nu)-\hat Q(\mu)-\hat Q(\nu)+\hat Q(0)=K(\mu,\nu)\quad \mbox{for all} \quad \mu,\nu \in \l\,.
\end{align}
A quadratic form is a quadratic refinement also satisfying $Q(n \mu)=n^2 Q(\mu)$ for all $n$.  Analogously, we define a quadratic refinement for finite bilinear form $(\sD,b)$: this is a function $\hat q: \sD\to \IQ/\IZ$ satisfying
\begin{align}
\label{QuRef}
\hat q(x+y)-\hat q(x)-\hat q(y)+\hat q(0)=b(x,y)\quad \mbox{for all}\quad  x,y\in \sD\,,
\end{align}
while a quadratic form $q$ for $b$ is a quadratic refinement that also satisfies $q(nx)=n^2q(x)$ (alternatively, it must satisfy $q(-x)=q(x)$ and $q(0)=0$). Note, the quadratic refinements that are not necessarily quadratic forms will be distinguished by a hat, like $\hat Q$ and $\hat q$, while quadratic forms are always unhatted, like $Q$ and $q$. The standard way to obtain a quadratic form for the even lattice $K$ is to define $Q(\mu)=K(\mu,\mu)/2$. If $K$ is odd, then $K(\mu,\mu)/2$ is not always an integer, and quadratic form refining $K$ does not exist. In order to deal with this problem, we define a quadratic refinement by picking a characteristic vector $W\in\l^{*}$, i.e., a vector satisfying 
\begin{align}
K(W,\mu)\equiv K(\mu,\mu)\;\mbox{mod}\;2\quad \mbox{for all}\quad \mu \in \l\;.
\end{align}
For any choice of $W$, $K(\mu,\mu-W)/2$ is a quadratic refinement of $K$. Moreover, we can shift this refinement by any constant. It is convenient to define the standard quadratic refinement for  $K$ odd as $\hat Q_W(\mu)=K(\mu,\mu\!-\!W)/2+K(W,W)/8$. This definition depends on the choice of $W$ in a very mild way. One checks that for any $W'\in \l^{*}$, there exists an element $\mu'\in \l$ such that $Q_{W+2W'}(\mu)=Q_{W}(\mu-\mu')$. Therefore, up to a shift of the argument, quadratic refinements defined above depend only on the class $[W]\in \l^{*}/2\l^{*}$. In the bosonic (even) case, this class can be chosen $0$, while in the spin (odd) case, a choice of non-trivial class $[W]$ is required. The refinement $\hat Q_{W}(\mu)$ descends to a refinement $\hat q_W(x)$ on the discriminant group in the full analogy with the discriminant bilinear form.

Another invariant of  the integral lattice $K$ is its signature $\sigma$ --  the difference between the numbers of positive and negative eigenvalues of $K$ diagonalized over $\IQ$. The Gauss-Milgram formula (see, e.g. \cite{Taylor, Hopkins2005}) states that the Gauss sum for the standard quadratic refinement $\hat q_W$, obtained from the lattice $(\l,K)$, satisfies 
\begin{equation}
\label{GaussMilgram}
\frac{1}{\sqrt{\sD}}\sum_{x\in \sD}e^{2\pi i \hat q_{W}(x)}=e^{2\pi i \sigma/8}\,.
\end{equation}

Let us define relations that will play a key role in our further discussion. Notice that a quadratic (bilinear) form precomposed with a group automorphism is another quadratic (bilinear) form. Since it is unreasonable to distinguish between isomorphic groups, we give the following definition of isomorphic quadratic (bilinear) forms on lattices and finite Abelian groups. 
\begin{definition}
Integral lattices $(\l,K)$ and $(\l',K')$ are isomorphic $(\l,K)\cong (\l',K')$ if there exists an isomorphism of free modules $F:\l\to \l'$ such that $K=K'\circ F$.
\end{definition}
\begin{definition}
\label{DefFinForms}
Finite bilinear forms $(\sD,b)$ and $(\sD',b')$ are isomorphic $(\sD,b)\cong(\sD',b')$ if there exists a group isomorphism $f:\sD\to \sD'$ such that $b=b'\circ f$.

Finite quadratic forms $(\sD,q)$ and $(\sD',q')$ are isomorphic $(\sD,q)\cong(\sD',q')$ if there exists a group isomorphism $f:\sD\to \sD'$ such that $q=q'\circ f$.
\end{definition}

There is an extra transformation which is allowed for quadratic refinements but not for quadratic forms.  Notice that the shift of the argument maps a quadratic form to a quadratic refinement. On the contrary, a shift of the argument of a quadratic refinement
 \begin{equation}
\label{shiftedQ}
    \hat q(x)\mapsto \hat q'(x)=\hat q(x-\delta)=\hat q(x)+\hat q(-\delta)+b(x,-\delta)-\hat q(0)
\end{equation}
gives another quadratic refinement refining the same bilinear form $b$ as $\hat q$ refines. Motivated by this observation, we introduce the following equivalence relation
 \begin{definition}
 \label{def:QuRefsEq}
Quadratic refinements $\hat q$ and $\hat q'$ defined on $\sD$  are equivalent $\hat q\;\sim \; \hat {q}'$ if there exists $\delta \in \sD$ such that $\hat q(x)=\hat q'(x+\delta)$ for all $x\in \sD$. The equivalence class of $\hat q$ is denoted by $\{\hat q\}$.
\end{definition}

With this new notation, we reformulate the previous result: given an integral lattice $(\l,K)$, the class of the quadratic refinement $\{q_{W}\}$ does not depend on the choice of the representative of the class $[W]\in \l^{*}/2\l^{*}$.

\subsection{Quantization of lattices}
In the previous section, we defined various invariants associated with integral lattices. Let us review the role those invariants play in the description of the quantum Abelian Chern-Simons-Witten theory. 

It was shown in \cite{Belov:2005ze} that the classical spin Chern-Simons theory, defined by the lattice $(\l,K)$, upon quantization becomes a theory which depends only on the quartet of lattice invariants $(\sD,\,b,\,\{\hat q_{W}\},\, \sigma\, \mbox{mod}\, 24)$ --- all topological invariants of manifolds assigned by this TQFT can be expressed through these four invariants. Note that \cite{Belov:2005ze} only classified torus CSW theory as a non-extended $\operatorname{2-3}$ spin TQFT, while torus CSW as a fully-extended $\operatorname{0-1-2-3}$ TQFT was analyzed in \cite{Freed:2009qp,Henriques}. 

Observe that keeping $b$ in the quartet is not necessary as $b$ can always be reconstructed from $\hat q$ according to \eqref{QuRef}. Nevertheless, we keep this redundant data in the quartet for a reason: many functions can refine the same finite form, and we will be finding refinements after fixing a finite form. In other words, we keep the redundancy to be aligned with the logic of this paper. 

The process of quantization of a toral CS theory can be seen as the following mapping:
\begin{align}
(\l,K)\mapsto ({\sD,\,b,\,\{\hat q_{W}\},\, \sigma\, \mbox{mod}\, 24})\,,
\end{align}
 called ``quantization of lattice'' in \cite{stirling2008abelian}. While the quantization is far from injective, it is surjective. It was shown by Wall in \cite{Wall1963QuadraticFO} that any finite form $(\sD,b)$ is a discriminant form for some integral lattice. Moreover, it was shown in Proposition 3.60 of \cite{stirling2008abelian} that there exists a unique class $[W] $ such that $\{q\}=\{q_{W}\}$. Therefore, any abstract quartet subject to \eqref{GaussMilgram} lifts to a (generically odd) lattice. If the quartet contains a refinement that is also a quadratic form, then such a quartet lifts to an even lattice. See \cite{Wang2020} for examples of such a reconstruction in the bosonic case. 

In summary, quantum toral CSW theory is defined by a quartet $ ({\sD,\,b,\,\{\hat q_{W}\},\, \sigma\, \mbox{mod}\, 24})$. The following isomorphism relation is natural.

\begin{definition}
Quartets $(\sD,\,b,\,\{\hat q\},\, \sigma\, \mbox{mod}\;24)$ and  $(\sD',\,b',\{\,\hat q'\},\, \sigma'\, \mbox{mod}\; 24)$ are isomorphic if there exists a group isomorphism $f:\sD\to \sD'$ such that $\{\hat q\}=\{\hat q'\circ f\}$,  $b=b'\circ f$, and $\sigma\equiv\sigma'$ mod $24$.  
\end{definition}

In the next section, we define the action of time-reversal on classical data and formulate the condition for quantum and classical time-reversal invariance.

\subsection{Time-reversal transformation} 
\label{Time-reversalDef}
Classically, a time-reversal transformation ($\mathsf{T}$-transformation) acts on the CS functional \eqref{3dCS} by changing the sign of the volume form. Equivalently, the $\mathsf{T}$-transformation acts as a negation of the bilinear form:
\begin{align}
    \mathsf{T}: (\l,K)\mapsto (\l,-K)\,.
\end{align}
The $\mathsf{T}$-transformation can be accompanied with a lattice automorphism, and, in some cases, the former can be undone with the latter.
\begin{definition}
The classical CS functional defined through the (isomorphism class) of lattices $(\l,K)$ is \underline{classically} time-reversal invariant ($\T$-invariant) if $(\l,K)\cong (\l,-K)$ as integral lattices.
\end{definition}
 If $(\sD,\,b,\,\{\hat q\},\, \sigma\, \mbox{mod}\, 24)$ is the quantization of $(\l,K)$, then $(\sD,-b,\{-\hat q\},  -\sigma\, \mbox{mod}\, 24)$ is the quantization of $(\l,-K)$. Along with the classical $\T$ invariance, we can introduce the quantum $\T$-invariance defined solely in terms of the lattice invariants. 
 
 \begin{definition}
 The quantum toral CS theory specified by the quartet $(\sD,\,b,\,\{\hat q\},\, \sigma\, \mbox{mod}\;24)$ is \underline{quantum} $\mathsf{T}$-invariant if this quartet is isomorphic to $(\sD,\,-b,\,\{-\hat q\},\, -\sigma\, \mbox{mod}\;24)$. Such a quartet is called $\T$-symmetric.
 \end{definition}
 A $\T$-symmetric quartet, in particular, consists of a finite form $(\sD,b)$ and a refinement $\hat q$ that are isomorphic to their negation. These terms deserve their own definition. 
 
  \begin{definition}
  \label{TformTref}
A finite bilinear form $(\sD,b)$ is a bilinear $\T$-form if $(\sD,b)\cong (\sD,-b)$. A quadratic refinement $\hat q$ on $\sD$ is a $\T$-refinement if there exists an automorphism $\gamma$ of $\sD$ such that $\{\hat q\circ \gamma\}=\{-\hat q\}$.
 \end{definition}
 Clearly, a quadratic $\T$-refinement refines a bilinear $\T$-form.

  \subsection{\texorpdfstring{$\sigma$}{Lg} mod 8 vs \texorpdfstring{$\sigma$}{Lg} mod 24}
  \label{RemarkOnSigma}
  Let us notice that the Gauss sum of any quadratic $\T$-refinement $\hat q$ is real: any $\T$-refinement satisfies the Gauss-Milgram constraint \eqref{GaussMilgram} with $2\sigma\equiv 0$ mod $8$. In the present paper, we find $\T$-refinements in two steps: we firstly find a finite form satisfying $(\sD,b)\cong (\sD,-b)$, and then we find all possible quadratic polynomials refining the $(\sD,b)$ found in the previous step, and whose Gauss sums are real. Assume $\hat q$ is a $\T$-refinement for $(\sD,b)$ satisfying $\eqref{GaussMilgram}$ with $\sigma\equiv 0$ mod $8$, then there exists only one $\T$-invariant quartet $(\sD,\,b,\,\{\hat q\},\, 0\, \mbox{mod}\, 24)$ containing that $\T$-form. Similarly, if $\hat q$ is a $\T$-refinement with $\sigma \equiv 4 $ mod $8$, then there is only one $\T$-invariant quartet $(\sD,\,b,\,\{\hat q\},\, 12\, \mbox{mod}\, 24)$ containing that refinement. Therefore, there is a $1-1$ correspondence between $\T$-refinements and $\T$-invariant quartets. 
 
The reconstruction theorem \cite{Belov:2005ze, stirling2008abelian} states that any quartet $(\sD,\,b,\,\{\hat q\},\, \sigma\, \mbox{mod}\, 24)$ is the image of the quantization map of some lattice. Clearly, the same theorem can proven about the data $(\sD,\,b,\,\{\hat q\},\, \sigma\, \mbox{mod}\, 8)$. Indeed, by lifting the $\sigma$ and adding $E_8$ or $-E_8$, we reduce the theorem to the one that is already proven. 

For example, we will find that $(\IZ/5,\; xy/5,\; \{2x^2/5\},\;12 \;\mbox{mod}\;24)$ is a $T$-invariant quartet. This quartet is a quantization of the root lattice $\mathbf{A}_4\oplus \mathbf{E}_8$. If we simply fix the refinement $2x^2/5$ on $\IZ/5$, it corresponds to $\sigma\equiv 4$ mod $8$ and it can be realized by the root lattice $\mathbf{A}_4$. Depending on the problem the reader is solving, we pick one of the two sets of data, and lift them to the corresponding lattices. In summary, there is no data lost when passing from $\sigma$ mod $8$ to $\sigma$ mod $24$, at least in the context of time-reversal theories.

Having in mind the preceding discussion, we will focus on finding $\T$-refinements assuming that $\T$-invariant quartets can always be recovered. 

\section{Finite Bilinear Forms}
\label{FinForms}
In this section, we review the classification of bilinear forms on finite Abelian groups and apply it to finding isomorphism classes of $\mathsf{T}$-invariant forms. As a general reference for this section, we recommend the book by Miranda
\cite{Miranda2009}. The main tool for the analysis of finite forms is localization to their prime and homogeneous subgroups. For a finite Abelian group $\sD$,  its $p$-subgroup $\sD_p$ (a.k.a. Sylow $p$-subgroup) is the maximal subgroup of $\sD$ consisting of elements whose order is a power of $p$. Any finite Abelian group is isomorphic to a direct sum of its $p$-subgroups: 
\begin{equation}
\label{Pdecomposition}
\sD \cong \bigoplus_{p} \sD_{p}\,.
\end{equation}
Further, each $\sD_p$ in this decomposition is isomorphic to a sum of homogeneous $p$-groups of exponent $k$ and  rank $d_p^{k}$:
\begin{equation}
\label{pgroup}
\sD_p \cong  \bigoplus_{k}(\IZ/p^{k})^{d_p^{k}}\,.
\end{equation}
Both splittings are orthogonal in the sense that any form $b$ on $\sD$ admits an orthogonal decomposition into the direct sum of forms on prime and homogeneous components:
\begin{align}
b\cong\bigoplus_{p}  b_{p}\,,\quad b_p\cong\bigoplus_{k} b_p^{k}.
\end{align}
Thus, any finite form admits a splitting $(\sD,b)\equiv \sum_p \sum_k ((\IZ/p^{k})^{d_p^{k}},b_p^k)$. 

\begin{definition}
 $\mathrm{bil}(\IZ)$ is the monoid of finite forms modulo the equivalence relation $\cong$ defined in \ref{DefFinForms} with orthogonal sum as the binary operation.
\end{definition} 
Since automorphisms of $\sD$ preserve its prime decomposition, $\mbox{bil}(\IZ)$ admits a decomposition into a direct sum of $p$-components $\mbox{bil}(\IZ)=\bigoplus_{p}\mbox{bil}(\IZ)_p$. A theorem by Wall \cite{Wall1963QuadraticFO} lists all generators of the semigroup $\mbox{bil}(\IZ)_p$ for any prime $p$. Let us review this theorem beginning with finite forms on cyclic groups.

  \paragraph{The finite forms on cyclic groups.} 
 The non-degeneracy condition implies that finite forms on $\IZ/p^r$ must has have the following form:
  \begin{align}
      b(x,y)= \alpha xy /p^r\,\quad \mbox{where} \quad x,y \in \IZ/p^r\,\quad \mbox{and}\quad \alpha\in (\IZ/p^r)^{\times}\,.
  \end{align}
Clearly, these forms are determined by the value $b(1,1)=\alpha /p^r$. Acting by the group automorphisms, we can get rid of the squares in $\alpha$. Thus, non-isomorphic forms are classified by $(\IZ/p^r)^{\times}/(\IZ/p^r)^{\times 2}$. Hensel's lemma tells that for odd prime $p$, 
quadratic residues mod $p^r$ are the same as quadratic residues mod $p$. Therefore, for $p$ odd, there are two isomorphism classes of forms: $b(1,1)=1/p^r$, denoted by $X_{p^r}$, and $b(1,1)=\theta/p^r$ where $\theta$ is any quadratic non-residue mod $p$, denoted by $Y_{p^r}$. The latter class does not depend on a specific choice of $\theta$ since all of them are related by a square mod $p^r$.

For $p=2$, the identity $m^2\equiv 1 \;\mbox{mod}\;8$ (and $m^2\equiv 1 \;\mbox{mod}\;4$), holding for any odd $m$, controls non-isomorphic forms on $\IZ/2^{r}$ for any $r$. As a consequence, the groups of squares $(\IZ/2)^{\times2}$, $(\IZ/4)^{\times2}$, and $(\IZ/8)^{\times2}$ are trivial, while  $(\IZ/2^r)^{\times2}\cong\{1,9,\dots \}$ for $r>3$. Therefore, $(\IZ/2^r)^{\times}/(\IZ/2^r)^{\times 2}$ contains not more than four elements for any $r$. This way, we listed all non-isomorphic forms on cyclic groups. It turns out that all finite forms are generated by forms on cyclic groups and forms on groups of the form $\IZ/2^r\oplus \IZ/2^r$. 

\begin{theorem}(Wall, 1963 \cite{Wall1963QuadraticFO})
1) The monoid $\mbox{bil}(\IZ)_p$ for odd $p$ is generated by forms on $\IZ/p^{r}$:
\begin{align}
X_{p^{r}}&\quad b(1,1)=1/p^{r}\,,\nonumber\\
Y_{p^{r}}&\quad b(1,1)=\theta/p^{r}\nonumber\,
\end{align}
where $\theta$ is any quadratic non-residue mod $p$.

2) The monoid $\mbox{bil}(\IZ)_2$ is generated by the following forms\footnote{The labels for the generators of $\mbox{bil}(\IZ)$ should not be confused with root lattices. The latter ones are denoted by bold letters like $\mathbf{A}_n$.}:
\begin{align}
&A_{2^{r}}\quad \mbox{on}\;\; \IZ/2^r\,,\;r\geqslant 1\;;\;b(1,1)=+1/2^{r}\,,\nonumber\\
&B_{2^{r}}\quad \mbox{on}\;\; \IZ/2^r\,,\;r\geqslant 2\;;\;b(1,1)=-1/2^{r}\,,\nonumber\\
&C_{2^{r}}\quad \mbox{on}\;\; \IZ/2^{r}\,,\;r\geqslant 3\;;\;b(1,1)=+5/2^{r}\,,\nonumber\\
&D_{2^{r}}\quad \mbox{on}\;\; \IZ/2^{r}\,,\;r\geqslant 3\;;\;b(1,1)=-5/2^{r}\,,\nonumber\\
&E_{2^{r}}\quad \mbox{on}\;\; \IZ/2^{r}\oplus \IZ/2^{r}\,,\;r\geqslant 1\;;\;b(e_i,e_j)=
\begin{cases}
0\quad &\mbox{if}\quad i=j\,,\nonumber\\
1/2^{r}\quad &\mbox{if} \quad i\neq j\,,\nonumber\\
\end{cases}\\
&F_{2^{r}}\quad \mbox{on}\;\; \IZ/2^r\oplus \IZ/2^r\,,\;r\geqslant 2\;;\;b(e_i,e_j)=
\begin{cases}
1/2^{r-1}\quad &\mbox{if}\quad i=j\,,\nonumber\\
1/2^{r}\quad &\mbox{if} \quad i\neq j\,,
\end{cases}
\end{align}
where $e_1$ and $e_2$ generate $\IZ/2^r\times \IZ/2^r$ in the last two cases.
\end{theorem}

Thus, for odd $p$, any form on the homogeneous group $(\IZ/p^{r})^{d}$ is isomorphic to 
\begin{align}
s_1X_{p^{r}}+ s_2Y_{p^r}\,
\end{align}
where $s_1+s_2=d$. However, this representation is not unique as there are relations among the generators. For odd $p$, the relations were found by Wall:
\begin{align}
2X_{p^r} = 2Y_{p^r}\,.
\end{align}
Due to this relation, the most general form on $(\IZ/p^r)^d$ is isomorphic to one of the two following types: 
\begin{align}
\label{OddType} 
\mbox{Type I}\;:\quad&(d-1)X_{p^r}+ Y_{p^r}\,,\\
\label{EvenType}
\mbox{Type II}:\quad&d X_{p^r}\,.
\end{align}
We call them normal forms of finite forms and choose them as the standard representatives of the given isomorphism class. Algebraically, Type I and Type II normal forms can be distinguished by Seifert's invariant $\chi_{p}^{r}(b)$ which we review in Appendix \ref{SeifertInv}. Type I forms correspond to $\chi_{p}^{r}(b)=-1$, while Type II forms correspond to $\chi_p^r(b)=+1$. 

\subsection{Normal forms for forms on 2-groups} 

We now turn our attention to 2-groups. 
\footnote{Here ``2-groups'' mean finite Abelian groups all of whose elements are torsion for some power of $2$, and should not be confused with the recently fashionable subject of two-step Postnikov extensions over $BG$ for some $G$.}
Relations among the generators of $\mbox{bil}(\IZ)_2$ were found by  Kawauchi and Kojima in \cite{Kawauchi1980}. These relations involve generators that are defined on groups with different exponents unlike the case of odd $p$. We do not provide the relations here, instead we jump to normal forms found by Miranda in \cite{Miranda1984}, and summarized in Table \ref{tab:table-name}.
\begin{table}[h!]
\begin{center}
    \begin{tabular}{ | c | c | l | c | c | c | }
    \hline
    Exponent $r$ & Rank $d$ & Normal form  & $\sigma_r$ & $\sigma_{r-1}$ & $\sigma_{r-2}$ \\ \hline
$1$ & $d\geqslant 1$ & $dA_{2}$ & $\infty$  & - & - \\
$1$ & $d$ even & $(d/2)E_{2}$ & $0$  & - & - \\
$2$ & $d\geqslant 1$ & $dA_{2^2}$ & $\infty$  & $d$ & - \\
$2$ & $d\geqslant 1$ & $(d-1)A_{2^2}+B_{2^2}$ & $\infty$  & $d+6$ & - \\
$2$ & $d\geqslant 2$ & $(d-2)A_{2^2}+2B_{2^2}$ & $\infty$  & $d+4$ & - \\
$2$ & $d\geqslant 3$ & $(d-3)A_{2^2}+3B_{2^2}$ & $\infty$  & $d+2$ & - \\
$2$ & $d$ even & $(d/2)E_{2^2}$ & $0$  & $0$ & - \\
$2$ & $d$ even & $((d/2)-1)E_{2^2}+F_{2^2}$ & $0$  & $4$ & - \\ 
$r\geqslant3$ & $1$ & $D_{2^r}$ & $\infty$  & $7$ & $3$ \\ 
$r\geqslant3$ & $2$ & $B_{2^r}+D_{2^r}$ & $\infty$  & $6$ & $2$ \\
$r\geqslant3$ & $3$ & $2B_{2^r}+D_{2^r}$ & $\infty$  & $5$ & $1$ \\
$r\geqslant3$ & $d\geqslant 1$ & $(d-1)A_{2^r}+C_{2^r}$ & $\infty$  & $d$ & $d+4$ \\
$r\geqslant3$ & $d\geqslant2$ & $(d-2)A_{2^r}+B_{2^r}+C_{2^r}$ & $\infty$  & $d+6$ & $d+2$ \\
$r\geqslant3$ & $d\geqslant3$ & $(d-3)A_{2^r}+2B_{2^r}+C_{2^r}$ & $\infty$  & $d+4$ & $d$ \\
$r\geqslant3$ & $d\geqslant4$ & $(d-4)A_{2^r}+3B_{2^r}+C_{2^r}$ & $\infty$  & $d+2$ & $d+6$ \\
$r\geqslant3$ & $d\geqslant1$ & $d A_{2^r}$ & $\infty$  & $d$ & $d$ \\
$r\geqslant3$ & $d\geqslant1$ & $(d-1)A_{2^r}+B_{2^r}$ & $\infty$  & $d+6$ & $d+6$ \\
$r\geqslant3$ & $d\geqslant2$ & $(d-2)A_{2^r}+2B_{2^r}$ & $\infty$  & $d+4$ & $d+4$ \\
$r\geqslant3$ & $d\geqslant3$ & $(d-3)A_{2^r}+3B_{2^r}$ & $\infty$  & $d+2$ & $d+2$ \\
$r\geqslant3$ & $d$ even & $(d/2)E_{2^r}$ & $0$  & $0$ & $0$ \\
$r\geqslant3$ & $d$ even & $((d/2)-1)E_{2^r}+F_{2^r}$ & $0$  & $4$ & $0$ \\
    \hline
    \end{tabular}
\end{center}
\caption{\label{tab:table-name}Normal forms of forms on $(\IZ/2^{r})^d$.}
\end{table}
The last three columns will be explained in a moment.
\begin{theorem}(Miranda, 1984)
\label{Theorem_Miranda}Let $\sD_2$ be a finite Abelian 2-group and $b$ be a form on $\sD_2$, then $b$ is isomorphic to a unique form expressed as a direct sum of the generators of $\mbox{bil}(\IZ)_2$ which is in normal form listed in Table \ref{tab:table-name}.
\end{theorem}

 Kawauchi and Kojima introduced the so-called signature invariants that distinguish isomorphism classes of forms on $2$-groups. We review the definition of these invariants in Appendix \ref{SignatureInv}. The signature invariants are valued in the semigroup $\overline{\IZ}_8=\IZ/8 \cup \infty$ -- the semigroup with $9$ elements $\{0,\,\dots,\,7,\,\infty\}$ where addition within $\IZ/8$ is the standard one and addition with $\infty$ is defined as
\begin{align}
x+\infty=\infty+x=\infty+\infty=\infty\,\quad \mbox{for all}\quad  x\in \IZ/8\,.
\end{align}
Any form $b$ on a $2$-group can be assigned a sequence of $\overline{\IZ}_8$-valued invariants $\{\sigma_{k}(b)\}_{k\in\IN}$. Elements of this formal sequence vanish for large enough $k$ for any finite Abelian $2$-group. The signature invariants are additive in the sense that $\sigma_k(b\oplus 
b')=\sigma_k(b)+\sigma_k(b')$, and define a semigroup homomorphism 
\begin{align}
\sigma_k:\mbox{bil}(\IZ)_2\to \overline{\IZ}_8\,.
\end{align} 
If $b$ is a form on a homogeneous 2-group with exponent $r$, then $\sigma_k(b)=0$ for $k>r$. The most interesting invariants $\sigma_r$, $\sigma_{r-1}$, and $\sigma_{r-2}$ for forms in normal form are given in Table \ref{tab:table-name} (all other invariants are not independent and can be deduced from those three). The general expression can be found in Appendix \ref{SignatureInv}.
\begin{proposition}(Miranda, 1984)
\label{Proposition_Miranda}  Any two forms $b$ and $b'$ on the $2$-group $\sD_2$ are isomorphic if and only if they have the same signature invariants $\{\sigma_{k}(b)\}_{k\in\IN}$ and $\{\sigma_{k}(b')\}_{k\in\IN}$.
\end{proposition}

As a conclusion, we use Seifert's invariants to classify finite forms on $p$-groups for odd $p$ and the signature invariants in the case of $2$-groups, such that the isomorphism class of the finite form is fully determined by the sets of invariants $\{d_p^{k}\}_{k\in\IN}$, $\{\chi_{p'}^{k}(b)\}_{k\in\IN}$, and $\{\sigma_k(b)\}_{k\in\IN}$ where $p$ is any prime, and $p'$ is any odd prime. 

\subsection{Bilinear \texorpdfstring{$\T$}{Lg}-forms}
Our goal here is to determine the classes of $\T$-forms given by Definition \ref{TformTref}.
Let us notice that the semigroup $\overline{\IZ}_8$ admits an involutive automorphism: 
\begin{align}
*:\overline{\IZ}_8&\to \overline{\IZ}_8\,,\\
*:x&\mapsto *x=\begin{cases}
-x\,,\quad x\in \IZ/8\,,\\
\,\,\infty \,,\quad x=\infty\,.\\
\end{cases}
\end{align}
As follows from definition of the signature invariants, 
\begin{align}
    \sigma_k(-b)=*\sigma_k(b)\,.
\end{align}
This observation leads us to the following proposition.
\begin{proposition}
\label{Tlinking}
The finite form $(\sD,b)$ is a finite $\mathsf{T}$-form if and only if all of the following conditions are satisfied\footnote{Here, $\left(\frac{a}{p}\right)$ stands for the Legendre symbol: it is $0$ if $p$ divides $a$, $+1$ if $a$ is a quadratic residue mod $p$, and $-1$ if $a$ is a quadratic non-residue mod $p$.}:

1) $d_{p}^{k}\equiv \begin{cases} 0 \;\mbox{mod}\; 2, \quad  \mbox{for any odd prime }p\,,\\
1\;\mbox{mod}\; 2,\quad \mbox{for any odd prime }p\mbox{ satisfying } \left(\frac{-1}{p}\right)=1\,. 
\end{cases}$

2) $\sigma_k(b)\in \{0,4,\infty\}$,

for all integers $k\geqslant 1$.
\end{proposition}

The statement about odd $p$ can be obtained by applying the negation transformation to a form in normal form and bringing it back to the normal form. For $p=2$, the forms which are negation fixed points are precisely the forms having signature invariants equal to $0$, $4$ or $\infty$. As a result, all isomorphism classes of finite $\mathsf{T}$-forms are specified by the invariants from Proposition \ref{Tlinking} and can be presented by forms in the normal form.

\paragraph{Example 1.} 
\label{Example1}
For odd prime $p$ such that $\left(\frac{-1}{p}\right)=1$, or equivalently, $p \equiv 1$ mod $4$, any finite form on $(\IZ/p^r)^d$ is a $\mathsf{T}$-form.

\paragraph{Example 2.} 
\label{Example2}
It turns out that $((\IZ/2^r)^{4},4A_{2^r})$ is always a $\mathsf{T}$-form for any $r\geqslant 1$. 

\paragraph{Example 3.} 
\label{Example3}
Trivial examples are $(\IZ/2,A_{2})$, $(\IZ/2\times \IZ/2,E_{2})$, and $(\IZ/4\times \IZ/4,F_{4})$. 
\begin{definition}
  The monoid of isomorphism classes of $\T$-forms is denoted by $\mathsf{T}\mathrm{bil}(\IZ)$. We also denote by $\mathsf{T}\mathrm{bil}(\IZ)_p$  the monoid of $\T$-forms on $p$-groups.  
\end{definition}
 Clearly, $\mathsf{T}\mbox{bil}(\IZ)\subseteq \mbox{bil}(\IZ)$ and $\mathsf{T}\mbox{bil}(\IZ)=\bigoplus_{p} \mathsf{T}\mbox{bil}(\IZ)_p$. By direct inspection of Table \ref{tab:table-name}, we obtain the following proposition. 
\begin{proposition}
\label{T-generators}
1) For a prime $p$ and for all $r \geqslant 1$, the monoid $\mathsf{T}\mathrm{bil}(\IZ)_p$ is generated by $X_{p^r}$, $Y_{p^r}$ if $p \equiv 1$ mod $4$ and by $2X_{p^r}$, $X_{p^r}+Y_{p^r}$ if $p \equiv 3$ mod $4$.

2) The monoid of $\T$-forms on homogeneous $2$-groups is generated by the following forms in normal form:
\begin{align}
A_{2}\,;\quad E_{2^r}\,;\quad F_{2^{m}}\,;\quad 4A_{2^{m}}\,;\quad A_{2^{n}}+B_{2^{n}}\,;\quad  B_{2^{n}}+C_{2^{n}}\,;\quad 3A_{2^{n}}+C_{2^{n}}
\end{align}
where $r\geqslant 1$, $m\geqslant 2$, and $n\geqslant 3$.
\end{proposition}
In the next section, we will define and separate a class of $\mathsf{T}$-forms that are somewhat trivially $\mathsf{T}$-invariant, e.g. the forms from Example 3.

\subsection{The Witt group of bilinear forms}
\label{WittBilgroup}
 We say that the finite form $(\sD,b)$ is split if there exists a direct summand $\sD'$ of $\sD$ such that the orthogonal complement to $\sD'$ is exactly $\sD'$ (we assume the form on $\sD$ is non-degenerate). The semigroup $\mbox{bil}(\IZ)$ modulo split finite forms is an Abelian group called the Witt group and denoted by $W$. Correspondingly, we have Witt groups $W_p$ and $W_p^k$ of forms on $p$-groups and homogeneous $p$-groups of exponent $k$. The following propositions by Kawauchi and Kojima describe the split forms in terms of their invariants and calculate the corresponding Witt groups.
\begin{proposition}(Kawauchi and Kojima, 1980) \label{KandK}
1) A finite form $(\sD,b)$ is split if and only if 
\begin{align}
d_p^k= 0\; \mbox{mod}\;2\quad \&\quad \chi_{p'}^k(b)=(-1)^{d_{p'}^{k}(p'-1)/4}\quad \&\quad \sigma_k(b)=0\; \mbox{or}\; \infty
\end{align}
for all prime $p$, all odd prime $p'$, and $k\geqslant 1$.

2) $W\cong \bigoplus_p W_p$. When $p$ is odd, $W_p\cong \bigoplus W_p^k$ and, for each $k\geqslant1$, 
\begin{align}
W_p^{k}\cong \begin{cases}
\IZ/2\oplus \IZ/2\quad \mbox{if}\quad \left(\frac{-1}{p}\right)=+1\,,\\
\IZ/4 \quad \quad \quad \;\; \,\,\mbox{if}\quad \left(\frac{-1}{p}\right)=-1\,,
\end{cases}
\end{align}
$W_2$ is isomorphic to a direct sum of infinite copies of $\IZ/2$, while $W_2^1\cong \IZ/2$, $W_2^2\cong \IZ/8$, and, for each $k\geqslant 3$, $W_2^{k}\cong \IZ/8\oplus \IZ/2$. 
\end{proposition}

Notice that for $p=2$, there is mixing among forms on groups with different exponents, and $W_2$ is not a direct sum of $W_2^k$ with different $k$'s. Matching Proposition \ref{Tlinking} and \ref{KandK}, we observe that the monoid of split forms is a sub--monoid of finite $\mathsf{T}$-forms. Analogously to the standard Witt group, we introduce the Abelian monoid of $\mathsf{T}$-forms modulo split forms and call it $\mathsf{T}$-Witt group, denoted by $\mathsf{T}W$. The subgroups $\mathsf{T}W_p$ and $\mathsf{T}W_p^{k}$ are defined accordingly. 

\begin{proposition}
\label{TWitt}
$\mathsf{T}W\cong \bigoplus_p \mathsf{T}W_p$. When $p$ is odd, $\mathsf{T}W_{p}\cong \bigoplus_k \mathsf{T}W_p^{k}$, and for each $k\geqslant 1$: 
\begin{align}
\mathsf{T}W_p^{k}\cong \begin{cases}
\IZ/2\oplus \IZ/2\quad \mbox{if}\quad \left(\frac{-1}{p}\right)=+1\,,\\
\IZ/2 \quad \quad \quad\;\; \,\,\mbox{if}\quad \left(\frac{-1}{p}\right)=-1\,.
\end{cases}\,
\end{align}
$\mathsf{T}W_2^1\cong \IZ/2$, $\mathsf{T}W_2^2\cong \IZ/2$ and, for each $k\geqslant 3$, $\mathsf{T}W_2^{k}\cong \IZ/2\oplus \IZ/2$.
\end{proposition}
\begin{proof}
A sketch of the proof is the following. The $W_p^k$ for $p\equiv 1$ mod $4$ can be presented as $\IZ/2\oplus \IZ/2\cong \la  X_{p^k}, Y_{p^k}\,|\,2X_{p^k}=2Y_{p^k}=0\ra$: all of the generators are $\mathsf{T}$-forms. The $W_p^k$ for $p\equiv 3$ mod $4$ can be presented as $\IZ/4\cong \la  X_{p^k}, Y_{p^k}\,|\, X_{p^k}+Y_{p^k}=0\ra$: the only $\mathsf{T}$-form among the generators is $2X_{p^r}$. For $p=1,k=1$, the only generator of the Wit group is $A_1$, for $p=2,k=2$, the generator is $A_4$, and for $p=2,k\geqslant 3$, the generators are $A_{2^k}$ and $B_{2^k}\oplus C_{2^k}$. Among them, the $\mathsf{T}$-forms are $A_2$, $4A_4$, and $B_{2^k}\oplus C_{2^k}$.
\end{proof}
For reference, let us collect presentations for $W$ and $\T W$. 
 \begin{align}
&W_{p}^{k}\cong \IZ/2\oplus \IZ/2\cong \la  [X_{p^k}], [Y_{p^k}]\,|\,[2X_{p^k}]=[2Y_{p^k}]=0\ra\,,\quad\quad\quad\quad\quad\quad\quad\quad\;\;\; p\equiv 1\;\mbox{mod}\;4\,,\;k\geqslant 1 \nonumber\\
&W_{p}^{k}\cong \IZ/4\cong  \la  [X_{p^k}], [Y_{p^k}]\,|\, [X_{p^k}+Y_{p^k}]=0\,,\;2[X_{p^k}]=2[Y_{p^k}] \ra\,,\quad\quad\quad\quad\quad \quad\, p\equiv 3\;\mbox{mod}\;4\,,\;k\geqslant 1  \nonumber\\
&W_{2}^{1}\cong \IZ/2 \cong  \la  [A_2]\,|\,2 [A_2]=0\ra \nonumber\\
&W_{2}^{2}\cong \IZ/8 \cong  \la  [A_4]\,|\,8 [A_4]=0\ra \nonumber\\
&W_{2}^{n}\cong \IZ/8\oplus \IZ/2 \cong  \la  [A_{2^n}],\;[A_{2^n}+D_{2^n}]\,|\,8 [A_{2^n}]=0,\;2[A_{2^n}+D_{2^n}]=0\ra\,,\quad n\geqslant 3 \nonumber\\
&\T W_{p}^{k}\cong \IZ/2\oplus \IZ/2\cong \la  [X_{p^k}], [Y_{p^k}]\,|\,2[X_{p^k}]=2[Y_{p^k}]=0\ra\,,\quad\quad\quad\quad\quad\quad\quad\quad\; p\equiv 1\;\mbox{mod}\;4\,,\;k\geqslant 1 \nonumber\\
&\T W_{p}^{k}\cong \IZ/2\cong   \la  2[X_{p^k}]\,|\,4 [X_{p^k}]=0 \ra\,,\quad\quad\quad\quad\quad \quad\quad\quad\quad\quad\quad\quad\quad\quad\quad\quad\quad p\equiv 3\;\mbox{mod}\;4\,,\;k\geqslant 1  \nonumber\\
&\T W_{2}^{1}\cong \IZ/2 \cong  \la  [A_2]\,|\,2[ A_2]=0\ra \nonumber\\
&\T W_{2}^{2}\cong \IZ/8 \cong  \la  [A_4]\,|\,8 [A_4]=0\ra \nonumber\\
&\T W_{2}^{n}\cong \IZ/2\oplus \IZ/2 \cong  \la  4[A_{2^n}],\;[A_{2^n}+D_{2^n}]\,|\,8 [A_{2^n}]=2[A_{2^n}+D_{2^n}]=0\ra\,,\;\;\;\; n\geqslant 3 \nonumber\\
\end{align}

Notice that the $\T$-Witt group is a subgroup of $W$ consisting of elements of order $2$. This is not a coincidence, as the following theorem demonstrates.
\begin{theorem}
\label{WittTforms}
If $(\sD,b)$ is a $\T$-form, then its class in the Witt group of bilinear forms $W$ is an element of order $2$. Any element of order $2$ in $W$ lifts to a $\T$-form.
\begin{proof}
 We consider $p=2$ and $p\neq2$ cases separately. 

1) $p=2$. Consider a $\T$-form $(\sD_2,b_2)$ on a $2$-group having invariants $\{d_2^{k}\}_{k=1}$ and $\{\sigma_k(b_2)\}_{k=1}$. According to Proposition \ref{Tlinking}, the signature invariants take values $0$, $4$, or $\infty$. Therefore, $(\sD_2\oplus \sD_2,b_2\oplus b_2)$ has signature invariants $0$ or $\infty$ and is necessarily split. For the second statement, let $(\sD_2,b_2)$ represent an order $2$-element in $W$: the signature invariants must satisfy $2\sigma =0$ or $\infty$. The multiplication table of $\bar \IZ_8$ shows that $(\sD_2,b_2)$ is necessarily a $\T$-form. 

2) $p\neq 2$. Let us fix a $\T$-form $((\IZ/p^k)^{d_{p}^{k}},b_p)$ on the homogeneous group of exponent $k\geqslant 1$. There are two subcases: $p\equiv 1$ mod $4$ and $p\equiv 3$ mod $4$.

$p\equiv 1$ mod $4$: $d_p^k$ and $\chi_{p}^k(b_p)$ are arbitrary. Multiplication by two makes $d_p^{k}$ even and $\chi_{p}^k(b_p)=1$. This is necessarily a split form according to Proposition \ref{KandK}.

$p\equiv 3$ mod $4$: $d_p^{k}$ is even, $\chi_p^k(b_p)$ is arbitrary. After multiplication by two, $d_p^{k}\equiv $ mod $4$, and $\chi_p^{k}(2b_p)=1$ -- this form must be split. 

 For the second statement of this theorem, we consider a finite form whose invariant $d_{p}^{k}$ satisfies $1=(-1)^{2d_{p}^{k}(p-1)/4}$. This holds for any $d_{p}^{k}$ if $p\equiv 3$ mod $4$, and for even $d$ if $p\equiv 3$ mod $4$. This way, we obtain all $\T$-forms on homogeneous $p$-groups with exponent $k$.
\end{proof}
\end{theorem}

\section{Quantum Data}
\label{QuantumData}
 In this section, we assemble the data defining quantum $\mathsf{T}$-invariant toral CS theories from the finite $\mathsf{T}$-forms. We fix the finite $\mathsf{T}$-form $(\sD,b)$ and try to find all possible $\T$-refinements refining this form.
 
 Let us begin with a review of the relations between finite bilinear and quadratic forms. A quadratic form $q$ refines the bilinear form $b$ on $\sD$ if $\eqref{QuRef}$ holds and $q(nx)=n^2q(x)$ for all $x\in \sD$. In particular, $2q(x)=b(x,x)$ for all $x\in \sD$. In general, this equation admits multiple solutions since the forms are valued in the torsion group $\IQ/\IZ$. However, multiple quadratic forms refining the same bilinear form are often isomorphic. It is a classical result  \cite{Wall1963QuadraticFO,DURFEE1977} that generators of $\mbox{bil}(\IZ)$ can be refined by not more than two quadratic forms up to isomorphism. The correspondence between the isomorphism classes of finite bilinear and finite quadratic forms is one-to-one \cite{Wall1963QuadraticFO,DURFEE1977} for forms on $p$-groups with odd $p$ and for forms on $2$-groups of exponent greater than $2$; otherwise the correspondence is one-to-two. 
 
 \begin{definition}
 $\mathrm{quad}(\IZ)$ is the monoid of finite quadratic forms modulo the equivalence relation $\cong$ defined in \ref{DefFinForms} with orthogonal sum as the binary operation.
\end{definition} 
Let us introduce a notation for the standard representatives of the generators of $\mathrm{quad}(\IZ)$:
 \begin{align}
 \label{quadformgenerators}
&\lambda_{p^{r}}^{+ 1}\quad \mbox{on}\;\; \IZ/p^r\,,\;r\geqslant 1\;;\;q(x)= ux^2/p^{r+1}\,,\nonumber\\
&\lambda_{p^{r}}^{-1}\quad \mbox{on}\;\; \IZ/p^r\,,\;r\geqslant 1\;;\;q(x)=vx^2/p^{r+1}\,,\nonumber\\
&\omega_{2^{r}}^{\pm 1}\quad \mbox{on}\;\; \IZ/2^r\,,\;r\geqslant 1\;;\;q(x)=\pm x^2/2^{r+1}\,,\nonumber\\
&\omega_{2^{r}}^{\pm 5}\quad \mbox{on}\;\; \IZ/2^{r}\,,\;r\geqslant 2\;;\;q(x)=\pm 5x^2/2^{r+1}\,,\nonumber\\
&\phi_{2^{r}}\quad \mbox{on}\;\; \IZ/2^{r}\oplus \IZ/2^{r}\,,\;r\geqslant 1\;;\;q(x,y)=xy/2^r\,,\nonumber\\
&\psi_{2^{r}}\quad \mbox{on}\;\; \IZ/2^r\oplus \IZ/2^r\,,\;r\geqslant 1\;;\;q(x,y)=(x^2+xy+y^2)/2^r
\end{align}
where $p$ is an odd prime, $\left( \frac{2u}{p}\right)=1$, and
$\textstyle{\left( \frac{2v}{p}\right)}=-1$.
 
 The correspondence between the generators of $\mathrm{bil}(\IZ)$ and $\mathrm{quad}(\IZ)$ is the following: $X_{p^r}\leftrightarrow \lambda_{p^{r}}^{+1}$, $Y_{p^r}\leftrightarrow \lambda_{p^{r}}^{-1}$, $E_{2^n}\leftrightarrow \phi_{2^n}$, $F_{2^n}\leftrightarrow \psi_{2^n}$, $A_{2^m}\leftrightarrow \omega_{2^m}^{+1}$, $B_{2^m}\leftrightarrow \omega_{2^m}^{-1}$, $C_{2^m}\leftrightarrow \omega_{2^m}^{+5}$, $D_{2^m}\leftrightarrow \omega_{2^m}^{-5}$ for all $r\geqslant 1$, $n\geqslant 2$, and $m\geqslant 3$. For small ranks, $A_2\leftrightarrow \{\omega_{2}^{+1},\omega_{2}^{-1}\}$, $A_4\leftrightarrow\{\omega_{4}^{+1},\omega_{4}^{+5}\}$, $B_4\leftrightarrow\{\omega_{4}^{-1},\omega_{4}^{-5}\}$, $E_2\leftrightarrow\{\phi_2,\psi_2\}$.

 \subsection{Quadratic \texorpdfstring{$\T$}{Lg}-forms}
 Let us now discuss quadratic $\T$-forms. First of all, we notice that a quadratic form $(\sD,q)$ refining the bilinear $\T$-form $(\sD,b)$ need not not be $\T$-symmetric. For example, both $\omega_2^{+1}$ and $\omega_{2}^{-1}$ refining the $\T$-form $A_2$ are not $\T$-symmetric. This happens because the negation operation maps one isomorphism class to another $\T:\omega_2^{\pm 1}\mapsto \omega_2^{\mp 1}$. On the contrary, if $(\sD,b)$ is a bilinear form that can be refined by exactly one isomorphism class of a quadratic form, then any quadratic form refining $(\sD,b)$ is automatically $\T$-symmetric. 
 
The mismatch between the bilinear and quadratic $\T$-forms happens only for $2$-groups of exponent not greater than $2$. Therefore, our analysis will be complete if we work out groups of the small rank separately. 
 \begin{lemma}
 \label{lemma}
 A quadratic form 
 \begin{align}
q=a_2^{+1}\omega_2^{+1}+a_2^{-1}\omega_2^{-1}+a_4^{+1}\omega_4^{+1}+a_4^{+5}\omega_4^{+5}+b_4^{-1}\omega_4^{-1}+b_4^{-5}\omega_4^{-5}
\end{align}
where $a_2^{+1}, a_2^{-1}, a_4^{+1}, a_4^{+5}, b_4^{-1}, b_4^{-5}$ are non-negative integers is $\T$-symmetric if and only if the following condition holds: 
\begin{align}
\label{Reality}
a_2^{+1}-a_2^{-1}+a_4^{+1}+a_4^{+5}-b_4^{-1}-b_4^{-5}\equiv 0 \;\mbox{mod}\;4\,.
\end{align}
Note, $a_2^{+1} \omega_2^{+1}$ is a short notation for $\bigoplus_{i=1}^{a_2^{+1}}\omega_{2}^{+1}$.
 \end{lemma} 
\begin{proof}
 Direct check or use Proposition 4.8 of \cite{Miranda2009}.
\end{proof}

Quadratic forms of type $\phi_2, \phi_4, \psi_2$, and $\psi_4$ are always $\T$-symmetric, so we did not include them in the lemma. In the previous section, we have used the invariants of Seifert and of Kawauchi and Kojima for the classification of finite bilinear forms. The latter invariants are certain Gauss sums of the associated quadratic form. 
\begin{definition}
Let $(\sD,q)$ be a non-degenerate quadratic form. We define its $n$-th Gauss sum $\tau_n(q)$ as 
\begin{align}
\tau_n(q)=\sum_{x\in \sD}e^{2\pi i n q(x)}\,.
\end{align}
\end{definition}
We can use the higher Gauss sums to characterize quadratic $\T$-forms directly, as the following theorem shows.
 
\begin{theorem}
\label{Th:quadTforms}
The finite quadratic form $(\sD,q)$ is $\T$-invariant if and only if the Gauss sums $\tau_n(q)$ are real for all $n=1,2,3,\dots\;$ .
\end{theorem}
\begin{proof}

In one direction, the proof is obvious: the Gauss sums depend only on the isomorphism class of the form while the $\T$-form and its negation must be isomorphic.

 To prove the opposite direction, we decompose $\sD$ into a direct sum $\sD=\bigoplus_{p,k}(\IZ/p^k)^{d_{p}^{k}}$, which is orthogonal with respect to $q$, and decompose $q$ into a sum of the generators \eqref{quadformgenerators}. Correspondingly, the $n$-th Gauss sum for $q$ is a product over the orthogonal summands. It follows from equation \ref{eq:lambda-Gauss-sum} and Table \ref{GaussSums2groups} below that the Gauss sums of $q$  will be real for all $n$ iff for all primes $p$, the Gauss sums of $q_p$ are real for all $n$.  Therefore, it suffices to analyze   the reality of the Gauss sums of $q_p$ for the distinct primes $p$. We first discuss $p$ odd and then $p=2$. 
 
For a prime $p>2$, there are only two generators of $\mbox{bil}(\IZ)_p$ and there is a 1-1 correspondence between the isomorphism classes of bilinear and quadratic forms. Let $n=p^{l}s$ where $s$ is coprime to $p$, then the $n$-th Gauss sum reads \cite{Miranda2009}:
 \begin{align}\label{eq:lambda-Gauss-sum}
     \tau_n(\lambda_{p^k}^{+1})=&\begin{cases}
     p^{k/2},\quad \; &l\geqslant k\,\\
     p^{l/2}\left(\frac{s}{p}\right)^{k-l}(-i)^{\frac{p-1}{2}},\quad \; &l\leqslant k\;\mbox{and}\;l-k\;\mbox{odd}\,\\
     p^{l/2},\quad \;\; &l< k\;\mbox{and}\;l-k\;\mbox{even}\,\\
     \end{cases}\\
          \tau_n(\lambda_{p^k}^{-1})=&\begin{cases}
     p^{k/2},\quad \; &l\geqslant k\,\\
     -p^{l/2}\left(\frac{s}{p}\right)^{k-l}(-i)^{\frac{p-1}{2}},\quad \; &l\leqslant k\;\mbox{and}\;l-k\;\mbox{odd}\,\\
     p^{l/2},\quad \; &l< k\;\mbox{and}\;l-k\;\mbox{even}\,\\
     \end{cases}
 \end{align}
Note, the Gauss sums $\tau_n(\lambda_{p^k}^{\pm 1})$ can have imaginary part only when $p$ divides $n$. If $\left(\frac{-1}{p}\right)=1$, then all the Gauss sums for both $\lambda_{p^k}^{\pm 1}$ are real. If $\left(\frac{-1}{p}\right)=-1$, then $\tau_{p^{k-1}}(\lambda_{p^k}^{\pm 1})$ is proportional to $i$.  Let $r$ be the largest exponent of $\sD_p$, then the imaginary part of $\tau_{p^{r-1}}$ is absent only if $d_p^r$ is even. We conclude that the reality of all Gauss sums for $(\sD_p,q_p)$ implies  Condition 1) of Proposition \ref{Tlinking}. In other words, all Gauss sums of $(\sD_p,q_p)$ are real if and only if it is a quadratic $\T$-form. 

Next, consider quadratic forms on $2$-groups. We provide the Gauss sums for the basic quadratic forms in Table \ref{Gausssums} \cite{Miranda2009}.
\begin{table}[h!]
\label{GaussSums2groups}
\begin{center}
    \begin{tabular}{  | c | c | c | c | }
    \hline
\multirow{2}{*}{$q$}  & \multicolumn{3}{|c|}{$\tau_{2^{l}s}(q)$} \\ \cline{2-4}
&$l<k$ & $l=k$ & $l>k$\\ \hline
$\omega_{2^k}^{t}$ & $2^{l/2}\left(\frac{2}{s}\right)^{k-l}e^{2\pi i s \frac{t}{8}}$ & $0$ & $2^{r/2}$  \\
  $\phi_{2^k}$ & $2^{l}$  & $2^{k}$ & $2^{k}$ \\
 $\psi_{2^k}$ & $(-1)^{k-l}2^{l}$  & $2^{k}$ & $2^{k}$ \\
    \hline
    \end{tabular}
\end{center}
\caption{\label{Gausssums}Gauss sums for the standard quadratic forms on 2-groups. Here, $l\geqslant0$ and $s$ is an odd natural number, and $t\in \{\pm 1, \pm 5\}$.}
\end{table}
Given the decomposition of $(\sD_2,q_2)$ into a direct sum of the generators \eqref{quadformgenerators}, we further split it into two parts: $(\sD_2,q_2)=(\sD_2^{r\leqslant 2},q_2^{r\leqslant2})\oplus(\sD_2^{r> 2},q_2^{r>2})$, where $\sD_2^{r\leqslant 2}$ is a subgroup of $\sD_2$ whose cyclic subgroups have exponent at most $2$ and $\sD_2^{r>2}$ is the rest. We analyze these two cases separately. 

First, consider $(\sD_2^{r\leqslant 2},q_2^{r\leqslant2})$: the reality of $\tau_1(q_2^{r\leqslant2})$ is equivalent to the reality of $\tau_3(q_2^{r\leqslant2})$ and equivalent to \eqref{Reality}. The reality of $\tau_1(q_2^{r\leqslant2})$ and $\tau_4(q_2^{r\leqslant2})$ does not imply any extra conditions. Clearly, $\phi_2,\phi_4,\psi_2$, and $\psi_4$ are $\T$-invariant and their Gauss sums are real. 

Next, we proceed to $(\sD_2^{r> 2},q_2^{r>2})$. We notice that the Gauss sums $\tau_{2^l}(q_2^{r\leqslant2})$ with $l\geqslant 1$ correspond to  the Kawauchi and Kojima invariants as of the associated bilinear form.  Recall, there is a 1-1 correspondence between the isomorphism class of $q_2^{r\leqslant2}$ and the associated bilinear form $b_2^{r>2}$ \cite{DURFEE1977}. Therefore, the Gauss sums and the Kawauchi and Kojima invariants are related as follows: 
\begin{align}
\tau_{2^l}(q_2^{r>2})=\begin{cases}
|\tau_{2^l}(q_2^{r>2})|\;e^{2\pi i \sigma_{l}(b_2^{r>2})}\,,\quad\mbox{if}\quad \sigma_l(b_2^{r>2})\in \IZ/8\subseteq \bar \IZ_8\,,\\
0\,,\quad\quad\quad\quad\quad\quad\quad\quad\quad\;\;\,\mbox{if}\quad\;\;\, \sigma_{l}(b_2^{r>2})=\infty\,.
\end{cases}
\end{align}
As a consequence, the reality of the higher Gauss sums $\tau_{2^l}(q_2^{r>2})$ for $l\geqslant 1$ is equivalent to $\sigma_{l}(b_2^{r>2})\in\{0,4,\infty\}$. It is easy to see that there are no extra conditions coming from the other Gauss sums. 

One might wonder about the logical possibility that there can be a non-T-invariant example where the Gauss sums of $\sD_2^{r\leqslant 2}$ and $\sD_2^{r> 2}$ are separately complex, but their product is always real. This option can be eliminated by considering the condition for the existence of a nonzero  imaginary part requiring  $l<k$, as stated in Table \ref{Gausssums}.  This completes the proof.
\end{proof}

One can notice that Theorem \ref{Th:quadTforms} can be proved by means of the methods developed in \cite{Miranda2009}, without making a reference to the bilinear $\T$-forms. However, we could not use those methods for studying quadratic $\T$-refinements which are not necessarily quadratic forms. 

When discussing the relation of T-invariant bilinear and quadratic forms it is useful to note the: 

 \begin{proposition}
\label{Prop1}
 Given an isomorphism class of a finite bilinear form $(\sD,b)$ in $\T\mbox{bil}(\IZ)$ and $\sigma\in \{0,4\}$, there exists at most one isomorphism class of a quadratic form $(\sD,q)$ in $ \T\mbox{quad}(\IZ)$ refining $(\sD,b)$ such that $\frac{\tau_1(q)}{|\tau_1(q)|}=e^{2\pi i \sigma/8}$.

\begin{proof}
The ambiguity exists only for $p=2$, $r\leqslant 2$. Then one can check it by hand. 
\end{proof}

\end{proposition}

There are examples of $\T$-symmetric bilinear forms that admit no $\T$-symmetric quadratic form. (A simple example is $A_2$.)  But, as we will see in Proposition \ref{UniqenessOfTref} below, there is always a $\T$-symmetric quadratic refinement. 

\subsection{The Witt group of quadratic forms.} 

In this section, we review the notion of the Witt group of quadratic forms, denoted by $\CW_{pt}$, and discuss the interplay between the Witt relation of quadratic forms and time-reversal symmetry. The Witt equivalence relation of quadratic forms answers the physically-motivated question of when two quantum Abelian CS TQFTs in three dimensions share a topological interface \cite{Saulina2011}.
One should not confuse the Witt group of bilinear forms $W$ defined in \ref{WittBilgroup} with the Witt group $\CW_{pt}$ we shall define here. The Witt group we define below is coarser than $W$ or its quadratic analog, while being more standard in the physics literature. In Theorem \ref{WittBilgroup} above we found a simple criterion for $\T$-invariance of a finite bilinear form expressed in terms of $W$. One might expect a similar statement for the $\T$-invariance of quadratic forms expressed in terms of $\CW_{pt}$, but the situation turns out to be more subtle. 

Let us review the standard notion of the Witt group of quadratic forms \cite{Drinfeld:2010} (compared  to the Witt group of bilinear forms $W$ defined in Section \ref{WittBilgroup}). We assume that quadratic and bilinear forms are non-degenerate unless specified otherwise. Let $(\sD,q)$ be a quadratic form, $\mathcal{H}\subseteq \sD$ be an isotropic subgroup, and $\mathcal{H}^{\perp}$ be the subgroup orthogonal to $\mathcal{H}$ with respect to the bilinear form refined by $q$. It is straightforward to check that $(\mathcal{H}^{\perp}/\mathcal{H}, q|_{\mathcal{H}^{\perp}/\mathcal{H}})$ is a non-degenerate quadratic form. 
We say that $(\sD,q)$ and $(\sD',q')$ are Witt-equivalent if there exist isotropic subgroups $\mathcal{H}\subseteq\sD$ and $\mathcal{H}'\subseteq\sD'$ such that $(\mathcal{H}^{\perp}/\mathcal{H}, q|_{\mathcal{H}^{\perp}/\mathcal{H}})\cong (\mathcal{H}^{'\perp}/\mathcal{H}', q|_{\mathcal{H}^{'\perp}/\mathcal{H}'})$. We denote the Witt equivalence class of $(\sD,q)$ by $[(\sD,q)]$. The monoid of isomorphism classes of quadratic forms modulo the Witt equivalence relation forms an Abelian group denoted by $\CW_{pt}$: the inverse to $[(\sD,q)]$ is $[\sD,-q]$. One can check that the following relations hold for all $k\geqslant 1$:
\begin{align}
\begin{gathered}
[\lambda_{p^{2k+1}}^{\pm1}]=[\lambda_{p}^{\pm1}]\,,\quad [\lambda_{p^{2k}}^{\pm1}]=[0]\,,\quad [\phi_{2^k}]=[0]\,,\quad [\psi_{2^{2k}}]=[0]\,,\quad [\psi_{2^{2k+1}}]=[\psi_{2}]\,,                          \\
[\omega_{2^{2k+1}}^{+1}]= [\omega_{2^{2k+1}}^{+5}]=[\omega_{2}^{+1}]\,,\quad [\omega_{2^{2k+1}}^{-1}]= [\omega_{2^{2k+1}}^{-5}]=[\omega_{2}^{-1}]\,,\\
 [\omega_{2^{2k+2}}^{+1}]=[\omega_{4}^{+1}]\,,\quad [\omega_{2^{2k+2}}^{+5}]=[\omega_{4}^{+5}]\,,\quad [\omega_{2^{2k+2}}^{-1}]=[\omega_{4}^{-1}]\,,\quad [\omega_{2^{2k+2}}^{-5}]=[\omega_{4}^{-5}]\,.
\end{gathered}
\end{align}
Clearly, any quadratic form is Witt equivalent to some anisotropic quadratic form, i.e., having no isotropic subgroups after reduction along the maximal isotropic subgroup (the choice of the maximal isotropic subgroup does not matter). Therefore, any Witt class can be represented by an anisotropic quadratic form. For completeness, we write down presentations for $\CW_{pt}$:
 \begin{align}
 \label{WittPtPresentations}
&( \CW_{pt})_p\cong \langle [\lambda_p^{+1}],[\lambda_p^{-1}]\,|\,2[\lambda_p^{+1}]=2[\lambda_p^{-1}]=0 \rangle \cong \IZ/2\oplus \IZ/2\,,\quad\quad\quad\quad\quad\,\;\; p\equiv 1\;\mbox{mod}\;4\,,\; \nonumber\\
&( \CW_{pt})_p\cong \langle [\lambda_p^{+1}],[\lambda_p^{-1}]\,|\, 2[\lambda_p^{+1}]=2[\lambda_p^{-1}],\; [\lambda_p^{+1}+\lambda_p^{-1}]=0] \rangle \cong \IZ/4\,,\quad\quad\;\; p\equiv 3\;\mbox{mod}\;4\,,\;  \nonumber\\
&( \CW_{pt})_2\cong \langle [\omega_2^{+1}], [\omega_2^{+1}+\omega_4^{-1}]\,|\,8[\omega_2^{+1}]=0,  2[\omega_2^{+1}+\omega_4^{-1}]=0\rangle\cong \IZ/8\oplus \IZ/2 \,.\;  \nonumber\\
\end{align}

A standard argument demonstrates that $\frac{1}{\sqrt{|\sD|}}\tau_1(q)=\frac{1}{\sqrt{|\mathcal{H}^{\perp}/\mathcal{H}|}}\tau_1(q|_{\mathcal{H}^{\perp}/\mathcal{H}})$ for any isotropic subgroup $\mathcal{H}\subseteq \sD$, see e.g., \cite{Miranda2009}. The argument works for any $\tau_n$ as long as $n$ is coprime to $|\sD|$. As a consequence,  $(\sD,q)$ represents the trivial element of $\CW_{pt}$ if and only if $\frac{\tau_n(q)}{|\tau_n(q)|}=1$ for all $n$ such that $gcd(n,\vert \sD \vert )=1$. Indeed, if $(\sD,q)$ contains a Lagrangian subgroup, then $\frac{\tau_n(q)}{|\tau_n(q)|}=\tau_n(0)=1$ for $gcd(n,|\sD|)=1$. To prove the opposite direction, we reduce any quadratic form to the anisotropic one, e.g. to the ones listed in \eqref{WittPtPresentations} and check that all non-zero anisotropic forms satisfy this property. A similar proof can be found in \cite{NG2019}.

Let us compare the Witt groups $W$ and $\CW_{pt}$. We notice that the class of $(\sD,q)$ in $\CW_{pt}$ contains much less information than the class of $(\sD,b)$ in $W$. We have proved that any bilinear form representing an order two element $W$ is $\T$ symmetric in Theorem \ref{WittTforms}. In sharp contrast, the condition that $(\sD,q)$ represents an order two element in $\CW_{pt}$ is \underline{not} sufficient for $\T$-symmetry. A good example, pointed out to us by T. Johnson-Freyd, of a $\CW_{pt}$-trivial quadratic form which is not $\T$-symmetric is $\lambda_{3}^{+1}+\lambda_{27}^{-1}$. This has a Lagrangian subgroup, so it is trivial in $\CW_{pt}$, but it is not $\T$-symmetric, as one can check by noting that the higher Gauss sums are not real. 

The $\T$-symmetry need not be preserved by the Witt equivalence. Indeed, if $(\sD,q)$ is a quadratic $\T$-form with the $\T$-symmetry realized by $\gamma \in Aut(\sD)$ and $\CH\subseteq \sD$ is an isotropic subgroup, then $\gamma (\CH)$ is isomorphic to $\CH$, but it is in general a different subgroup of $\sD$.

Any quadratic $\T$-form $(\sD,q)$ represents an order two element in $\CW_{pt}$ by definition: $[\sD\oplus \sD,q\oplus q]=[\sD\oplus \sD,q\oplus -q]=[0]$. On the other hand, any order two element of $\CW_{pt}$ can be presented by an anisotropic quadratic form which must be $\T$-symmetric as can be seen in \eqref{WittPtPresentations}. In principle, one should be able to define an analog of the Witt relation which would preserve the $\T$-symmetry, perhaps, by requiring that the isotropic subgroup is fixed by at least one time reversal. We leave this question for future work.

\subsection{Generalized isometries}
Let us conclude the discussion of quadratic $\T$-refinements by extending groups of isometries to graded groups which might negate the quadratic function. 

  For any quadratic form $q$ on a finite Abelian group $\sD$, we define its isometry group $Isom^{+}( q)$ as follows:
\begin{align}
\label{IsometriesDef}
Isom^{+}( q)=\{\gamma \in Aut(\sD)\,|\;q \circ \gamma=q\}\,.
\end{align}
The isometry groups of the standard quadratic forms are known, see, e.g. \cite{Miranda2009}:
\begin{align}
\label{IsoGr}
    &Isom^{+}(\omega_2^{\pm 1})=\{ \mbox{Id}\} \,;\quad Isom^{+}(\omega^{t}_{2^m})=\{ \pm \mbox{ Id }\} \,;\quad Isom^{+}(\lambda^{\pm1}_{p^r})= \{ \pm \mbox{ Id }\} \,;\nonumber \\
    &Isom^{+}(\phi_{2^r})=\bigg\{
\begin{pmatrix} u & 0\\
0 & u ^{-1}
\end{pmatrix},\begin{pmatrix} 0 & u\\
u^{-1} & 0 
\end{pmatrix}\,| u\in (\IZ/2^r)^{\times}\bigg\}\,;\\
    &Isom^{+}(\psi_{2^{m}})=\bigg\{
\begin{pmatrix} a & -c\\
c & a+c
\end{pmatrix},\begin{pmatrix} a & a+c\\
c & -a 
\end{pmatrix}\,| a,c\in \IZ/2^r\,, a^2+ac+c^2=1\bigg\}\,\nonumber 
\end{align}
where $t\in \{\pm 1,\pm 5\}$, $p$ is an odd prime, $r \geqslant 1$, and $m\geqslant 2$. Clearly, isometry groups for the most general quadratic forms can be obtained from the isometry groups above. 

Further, we can define the set $Isom^{-}(q)$ of automorphisms negating the given quadratic form 
\begin{align}
 Isom^{-}(q)=\{\gamma \in Aut(\sD)\,|\,q \circ \gamma=-q\}\,.
\end{align}
By definition, the set $Isom^{-}(q)$ is non-empty for a quadratic $\T$-form and empty otherwise. Unlike the set of isometries, $Isom^{-}(q)$ is not known explicitly: there is no general formula for the square roots of $-1$ even for cyclic groups. It is convenient to define the group $Isom^{\pm}(q)$ of generalized isometries as 
\begin{align}
 Isom^{\pm}(q)=\{\gamma \in Aut(\sD)\,|\,q \circ \gamma=\pm q\}\,.
\end{align}
For cyclic groups, we obtain the abstract groups of generalized isometries:
\begin{align}
    &Isom^{\pm}(\omega_2^{\pm 1})=\mbox{Id}\,;\; Isom^{\pm}(\omega^{t}_{2^{m}})\cong \IZ/4\,;\; Isom^{\pm}(\lambda^{\pm1}_{p^{r}})\cong \IZ/4
\end{align}
where notation is the same as in \eqref{IsoGr}. The $Isom^{\pm}(q)$ is $\IZ/2$-graded group and is the direct analog of the $\IZ/2$-graded group of unitary and anti-unitary symmetries allowed in the ordinary quantum mechanics by Wigner's theorem \cite{Geiko:2020mdx,Freed:2011aa,Freed:2012uu,Wigner1959}.

One of the main goals of this paper is to discuss the necessary conditions for the existence of quantum time-reversal symmetry. Assuming that the quantum Chern-Simons admits a $\T$-symmetry, any concrete realization of this symmetry requires a choice of an element $\gamma\in Isom^{-}(q)$. This choice is important for the anomaly partition function \cite{Lee:2018eqa}, or, if we wish to ``gauge'' this symmetry \cite{Barkeshli:2014cna}. 

\subsection{Quadratic \texorpdfstring{$\T$}{Lg}-refinements}
\label{Trefinements}
 We proceed to the discussion of the most general quadratic $\T$-refinement for a given bilinear $\T$-form. In this section we solve the problem of finding the set of equivalence classes of quadratic refinements $\{\hat q\}$ satisfying \eqref{QuRef}, which are $\T$-invariant in the sense of Definition \ref{TformTref}, refining the fixed finite bilinear $\T$-form $(\sD,b)$.
 
 We begin with stripping off the constant from the quadratic refinement by defining the quadratic enhancement of $b$, a function $\tilde q(x)$ satisfying
 \begin{align}
     \tilde{q}(x+y)-\tilde{q}(x)-\tilde{q}(y)=b(x,y)\,.
 \end{align}
 The quadratic enhancement is related to a quadratic refinement $\hat q(x)$ through $\tilde q(x)=\hat q(x)-\hat q(0)$. If $\tilde q(x)$ is any enhancement of $b$, then any other enhancement can be obtained by adding a group homomorphism $l:\sD\to \IQ/\IZ$, see e.g. \cite{Taylor}. Since $b$ is non-degenerate, any such a group homomorphism can be presented as $l(x)=b(x,\Delta)$ for some $\Delta\in \sD$. To this end, we bring back the constant and find that the most general quadratic refinement of $b$ has the form $\hat q(x)=q(x)+b(x,\Delta)+\hat q(0)$ where $q(x)$ is any quadratic form refining $b$, $\Delta\in \sD$, and $\hat q(0)\in \IQ/\IZ$ --all subject to the Gauss-Milgram constraint \ref{GaussMilgram}.
 
 Recall, the equivalence class of a quadratic refinement $\{\hat q\}$, consists of all refinements related by a shift of the argument. The shift of the argument acts on a refinement in the following way $\hat q(x)\mapsto \hat q(x+\delta)=\hat q(x)+\hat q(\delta)+b(x,\delta)-\hat q(0)$. Therefore, any equivalence class $\{\hat q\}$ can be presented by a refinement of the form $\hat q(x)=q(x)+c$ where $q(x)$ is any quadratic form refining $b$ and $c=\alpha/8$ for some $\alpha \in \IZ/8$ (the last condition follows from the fact that $\hat q(x)$ must satisfy the Gauss-Milgram constraint for some $\sigma$).
 
Finally, let us assume that $\{\hat q\}$ refines $(\sD,b)$ and require it to be a $\T$-refinement according to Definition $\ref{TformTref}$: there must exist a group automorphism $\gamma\in Aut(\sD)$ such that $\{\hat q\circ \gamma\}=\{-\hat q\}$ (the same automorphism will negate $b$). We present $\{\hat q\}$ by the refinement of the form $q(x)+c$ and apply the automorphism. There are two exceptional bilinear forms, $A_2$ and $E_2$, which are $\T$-symmetric identically, i.e., the negating isomorphism is the identity. $E_2$ can be refined by a quadratic $\T$-form $\phi_2$ and $\psi_2$, while there is no quadratic $\T$-form refining $A_2$. Since the class of $\{\hat q_{A_2}\}$ can be presented by $\hat q_{A_2}(x)=x^2/4+\alpha/8$ with some $\alpha\in \IZ/8$, we immediately find that $A_2$ admits two classes of quadratic $\T$-refinements $\{x^2/4-1/8\}$ and $\{x^2/4+3/8\}$. 
For all other bilinear forms, we find that for $\{\hat q \}$, presented by $q(x)+c$, the necessary and sufficient condition of $\T$-invariance is that $q(x)$ is quadratic $\T$-form and $c\in \{0,1/2\}$.

Let us summarize our findings. The bilinear $\T$-form of the type $A_2$ admits quadratic $\T$-refinements of the forms $\{\omega_2^{+1}-1/8\}$ and by $\{\omega_2^{+1}+3/8\}$. Bilinear $\T$-forms of other types  admit  $\T$-refinements presented by $q(x)$ and $q(x)+1/2$ where $q(x)$ is a quadratic $\T$-form refining $b$. 

We constructively proved that the class of a quadratic refinement $\{\hat q\}$ is $\T$-symmetric only if all the higher Gauss sums of any representative $\hat q $ of this class are real. Since the necessary condition holds automatically, we obtain that an analog of Theorem $\ref{Th:quadTforms}$ holds verbatim for quadratic refinements. The discussion above can be summarized by the following proposition.

\begin{proposition}
\label{UniqenessOfTref}
Given an isomorphism class of a bilinear $\T$-form $(\sD,b)$ and $\sigma\in \{0,4\}$, there exists exactly one equivalence class of a quadratic $\T$-refinement $\{\hat q\}$ which refines $(\sD,b)$ and such that $\frac{\tau_1(\hat q)}{|\tau_1(\hat q)|}=e^{2\pi i \sigma/8}$.
\end{proposition}

\section{Lattices}\label{Lattices}

At the classical level, any concrete toral CS functional is specified by a lattice $(\l,K)$. In this section, we discuss the lifting of quantum data corresponding to $\mathsf{T}$-invariant quantum CS theories to lattices. Clearly, isomorphic lattices have isomorphic discriminant forms and quadratic refinements. In general, the interplay between isomorphism classes of finite forms and lattices is described by the following theorem.

\begin{theorem}(\cite{Kneser1953QuadratischeFU},\cite{DURFEE1977})
\label{StableLattices}
Two lattices $(\Lambda,K)$ and $(\Lambda',K')$ have isomorphic discriminant forms if and only if they are stably equivalent, that is $K\oplus U\cong K'\oplus U'$ where $U$ and $U'$ are unimodular lattices, even if $K$ and $K'$ are even, and odd if $K$ and $K'$ are odd.
\end{theorem}
Thus, finite $\mathsf{T}$-forms lift to lattices that are classically $\mathsf{T}$-invariant up to the addition of a unimodular lattice. In principle, we can describe lattices corresponding to $\mathsf{T}$-invariant quartets in terms of their $p$-adic reductions, also known as the genus. However, there is a more constructive description in terms of lattice embeddings. 

\subsection{Lattice embeddings}

Let us review some of the results obtained by Nikulin in \cite{Nikulin1980}. An embedding of lattices $\phi:\Lambda\hookrightarrow M$ is called primitive if the quotient $M/\Lambda$ is a free $\IZ$-module, i.e., another lattice. Note that, in general, embeddings of lattices need not be primitive. For example, $\l/2\l$ is a pure torsion, while for primitive embeddings, $M/\Lambda$ must be torsion-free.
\footnote{
In practice, the primitiveness of the given embedding can be determined by a criterion found by Sert{\"o}z \cite{sertoz2005singular} (Theorem 6): a lattice embedding is primitive if and only if the greatest common divisor of the maximal minors of the embedding matrix with respect to any choice of basis is 1. This criterion was rederived in \cite{Levin2012}. }
For a primitive embedding $\Lambda\hookrightarrow M$, we can consider the orthogonal complement $\Lambda^{\perp}\equiv (\Lambda)_{M}^{\perp}$ with respect to the bilinear form of $M$. We say that two lattices $\Lambda$ and $\Lambda'$ are perpendicular $\Lambda \!\perp \! \Lambda'$ if there exists a unimodular lattice $M$ and a primitive embedding of $\Lambda$ into a $M$ such that $\Lambda^{\perp}\cong \Lambda'$.

\begin{proposition}(Nikulin)
\label{NikulinEmbedding}
 Let $(\sD,b)$ and $(\sD',b')$ be the discriminant forms of $(\Lambda,K)$ and $(\Lambda',K')$ respectively. Then: 
 
 1) A primitive embedding of $\l$ into a unimodular lattice $M$ such that $\Lambda^{\perp}\cong \Lambda'$ is determined by an isomorphism $\gamma: \sD\to \sD'$ for which $b'\circ \gamma=-b$.

2) $(\l,K)\!\perp \!(\l',K')$ if and only if $(\sD,b)\cong(\sD', -b')$.
\end{proposition}

It is straightforward to see that any lattice is orthogonal to its negation: $(\l,K)\!\!\perp \!\!(\l, -K )$ trivially holds. However, the lattice need not be perpendicular to itself. As follows from the proposition above, $(\l,K)\!\perp\! (\l,K)$ if and only if $(\sD, b)$ is a T-form. This gives us an alternative description of $\mathsf{T}$-forms: these are discriminant forms of self-perpendicular lattices.

\paragraph{Remark.} Physically, the condition $\Lambda \!\perp \! \Lambda'$ for lattices means that there are three toral Chern-Simons theories: $CS_{\l,K}$, $CS_{\l',K'}$, and $CS_{M,U}$. Since $(M,U)$ is a unimodular lattice, $CS_{M,U}$ is a ``trivial'' theory. Note, the primitivity of the embedding is crucial; otherwise, $M/\l$ has a torsion and does not define a toral gauge theory. The ``trivial'' theory has no anyons, yet it is still a toral gauge theory with the lattice of charges and we can pick off charges constituting the non-trivial theory $CS_{\l,K}$ such that the quotient theory is $CS_{\l',K'}$. In its turn, if $CS_{\l,K}$ is $\T$-invariant upon quantization, then it can be embedded into a trivial theory such that the quotient is the same theory. There is a close relation to the folding trick \cite{Saulina2011}, where folded theory corresponds to $(\l\oplus \l, K\oplus (-K))$. Here, the pair $CS_{\l,K}$ and $CS_{\l,-K}$ is the simplest example of theories embeddable in the ``trivial'' theory.

\subsubsection{Examples of rank 2}
As an example, let us demonstrate a primitive self-orthogonal embedding of a rank $2$ lattice in a rank $4$ unimodular lattice. Let $e_1,e_2$ be a basis for $\l$ in which the matrix of $K$ is given by
\begin{align}
K_{m,n}=\begin{pmatrix} m^2+n^2 & 0\\
0 & -1
\end{pmatrix}\,\quad  \mbox{where $m$ and $n$ are coprime non-zero integers}\,.
\end{align}
Let $\IZ^{2,2}$ be the odd unimodular lattice of rank $4$ with zero signature: we pick a basis $f_1,\dots ,f_4$ in which the bilinear form is  $\mbox{diag}(+1,+1,-1,-1)$. Consider the following embedding of $(\l,K_{m,n})$ into $\IZ^{2,2}$:
\begin{align}
&e_1 = m f_1 +n f_2\,,\\
&e_2= f_3\,.
\end{align}
This embedding is primitive as long as $gcd(m,n)=1$, the quotient is torsion-free due to B\'ezout's identity. For this embedding, the orthogonal complement $(\Lambda)^{\perp}_{\IZ^{2,2}}$ is embedded into $\mathbb{Z}^{2,2}$ in the following way:
\begin{align}
&e_1^{\perp} = n f_1 -m f_2\,,\\
&e_2^{\perp}= f_4\;.
\end{align}
We see that $(\Lambda)^{\perp}_{\IZ^{2,2}}$ is isomorphic to $\l$, yet at the same time, proposition \eqref{NikulinEmbedding} tells that $b_{\l}\cong -b_{\l}$ and, consequently, $(\l,K_{m,n})$ is stably isomorphic to $(\l,-K_{m,n})$. As was shown in \cite{Delmastro:2019vnj}, $(\l,K_{m,n})$ is classically $\T$-invariant iff $m^2+n^2$ solves the negative Pell equation, i.e., there are integers $p$ and $q$ such that $(m^2+n^2)p^2-q^2=1$. 

This example of rank $2$ lattices was analyzed in \cite{Seiberg2016,Delmastro:2019vnj}. Moreover, Delmastro and Gomis prove that these lattices capture all $\mathsf{T}$-invariant lattices of the form $\mbox{diag}(n,-1)$. They also analyzed the isometry groups of such lattices. 

Let us check that the quantization of lattices we have constructed gives us $\mathsf{T}$-invariant quartets. We have $\sD\cong \IZ/(m^2+n^2)$. Prime factorization of $m^2+n^2$ for non-zero coprime integers $m$ and $n$ may contain $2$ with the power of $0$ or $1$ and any powers of primes of the form $p\equiv 1$ mod $4$.
\footnote{Let us prove this statement. Notice that prime divisors of $d=m^2+n^2$ will divide neither $m$, nor $n$. Otherwise, $m$ and $n$ would not be coprime. Then, $m$ and $n$ are invertible modulo any prime $p$ s.t. $p|d$ and there exist $l\neq 0$ s.t. $m\equiv ln$ mod $p$. This implies that $l^2+1\equiv 0 $ mod $p$ for any such $p$. Therefore, $-1$ is a square residue modulo any prime divisor of $d$.}
As expected, the discriminant form is automatically a $\mathsf{T}$-form. Depending on the parity of $m^2+n^2$, there are two inequivalent refinements for the discriminant form: 
\begin{align}
    \hat q(x)=\begin{cases} \frac{x^2}{2(m^2+n^2)}-\frac{1}{8}\,,\quad\quad\;\;\, m^2+n^2\quad \mbox{even}\,,\\
    \frac{x^2}{m^2+n^2}+\frac{m^2+n^2-1}{8}\,,\quad m^2+n^2 \quad\mbox{odd}\,.
    \end{cases}
\end{align}

\subsubsection{Examples of rank 4}
Consider the root lattices $\mathbf{A}_4$ and $\mathbf{D}_4$: they can be primitively embedded into $\mathbf{E}_8$ in an obvious way. It turns out that their orthogonal complements within $\mathbf{E}_8$ are equivalent to $\mathbf{A}_4$ and $\mathbf{D}_4$ respectively. For reference, we provide the concrete realizations below.  
\begin{align}
A_4=\begin{pmatrix} 2 & -1 & 0 & 0\\
-1 & 2 & -1 & 0\\
0& -1 & 2 & -1 \\
0 & 0 & -1 & 2 
\end{pmatrix}\,\quad \; 
A_4^{\perp}=\begin{pmatrix} 2 & 0& -1 & 0\\
0 & 2 & 0 & -5\\
-1& 0 & 2 & -5 \\
0 & -5 & -5 & 30 
\end{pmatrix}\\
D_4=\begin{pmatrix} 2 & -1 & 0 & 0\\
-1 & 2 & -1 & -1\\
0& -1 & 2 & 0 \\
0 & -1 & 0 & 2 
\end{pmatrix}\,\quad \; 
D_4^{\perp}=\begin{pmatrix} 12 & -6 & 2 & 0\\
-6 & 4 & -2 & 0\\
2& -2 & 2 & -1 \\
0 & 0 & -1 & 2 
\end{pmatrix}
\end{align}
We have the equivalences $\mathbf{A}_4 \cong \mathbf{A}_4^{\perp}$ and $\mathbf{D}_4\cong \mathbf{D}_4^{\perp}$. However, these lattices are not equivalent to their negations simply because they are positively-defined. Thus, these are examples of quantum $\mathsf{T}$-invariant lattices. Their quantization gives $(\IZ/5, xy/5, 2x^2/5,4)$ and $(E_2, xy/2, 4)$ respectively. The case of $\mathbf{D}_4$ was also found in \cite{Nikulin1980,Wang2013}.

\subsection{Classical vs quantum time-reversal}
In our approach to Abelian CS, the quantum time-reversal symmetry is much more common than the classical one. For the former, it suffices to find a $\mathsf{T}$-form, while for the latter, there must be an equivalence between the lattice and its negation. Clearly, a positive-definite lattice is never classically $\mathsf{T}$-invariant. Let $t_{+}$ and $t_{-}$ be the numbers of positive and negative eigenvalues of $K$ over $\IQ$. Then, we might expect $(\l,K)$ to be classically $\mathsf{T}$-invariant only if $t_{+}=t_{-}$. Let $l(\sD)$ be the minimal number of generators of $\sD$.
\begin{proposition}
\label{NikulinUniqueness}(Nukulin) An even lattice $(\l,K)$ with invariants $(t_{+},t_{-},q)$ is unique up to isomorphism if $t_{+}+t_{-}\geqslant 2 +l(\sD)$, $t_+\geqslant 1$, $t_{-}\geqslant 1$.
\end{proposition}
This proposition gives sufficient conditions for an even lattice to be $\mathsf{T}$-invariant. A classical Abelian CS theory defined by an even indefinite lattice $(\l,K)$ is classically $\mathsf{T}$-invariant if $t_{+}=t_{-}$, its discriminant form $(\sD,b)$ is a $\mathsf{T}$-form, $ q$ is a $\mathsf{T}$-refinement that is also a quadratic form, and $2t_{+}\geqslant 2 +l(\sD)$. This condition eliminates $A_2$, as for this $\mathsf{T}$-form, there does not exist a quadratic $\mathsf{T}$-form. We expect that a similar condition can be formulated for odd lattices, and leave this question for the future.

\section{Modular Tensor Categories}
\label{MTC:section}
We discussed the quantum data defining quantum Abelian Chern-Simons-Witten theory as a result of canonical quantization. Every (spin) CSW theory defines a 3d TQFT, known as a Witten-Reshetikhin-Turaev theory \cite{Witten1989} and, by the correspondence with 2d RCFT, it also defines a modular tensor category \cite{Moore:1988qv,Moore:1989vd}. In the spin case it should be a super-MTC. 
We shall view the (super-) MTC as the quantum data associated to a quantum CSW theory. 
In this section we briefly review how the quantum data of torus CSW theory fits into the categorical picture. 
We then define time-reversal invariance in terms of (super-) MTCs and consider some implications of this definition. 

We begin by setting the most necessary notation\footnote{See   Appendix \ref{WittAppendix} for quick definitions of other terms from category theory used here.}. Recall that a modular tensor category $\CC$ is, in particular, a non-degenerate braided fusion category \cite{Drinfeld:2010} equipped with a ribbon structure, see \cite{Rowell2018,Kong:2022cpy} for recent reviews.  The braided structure is a set of isomorphisms permuting objects in the monoidal product 
\begin{align}
\label{braiding}
    c_{x,y}: x\otimes y  \xrightarrow{\sim} y \otimes x \,,
\end{align}
while the ribbon structure is a set of isomorphisms $\theta_x:x\xrightarrow{\sim} x$, called twists, compatible with the braiding, i.e., satisfying the balancing condition:
\begin{align}
\label{balancing}
    \theta_{x\otimes y}=c_{y,x}\circ c_{x,y}(\theta_x\otimes \theta_y)
\end{align}
and such that $\theta_{\mathbf{1}}=id_{\mathbf{1}}$ where $\mathbf{1}$ is the monoidal unit.
\footnote{The physical origin of this condition is that, if $\phi_1, \phi_2, \phi_3$ are (holomorphic parts of) primary fields in 2d CFT, of conformal dimensions $h_1, h_2, h_3$, then the coefficient of $\phi_3(w)$ in the operator product of $\phi_1(z)$ and $\phi_2(w)$ depends on $z,w$ as a constant times $(z-w)^{h_3 - h_2 - h_1}$.}

Since we are interested in super-MTCs arising from (spin) Chern-Simons theories, all of the categories below are assumed to be unitary. 
The spherical structure allows us to define the trace, such that the modular S-matrix is defined through the braiding as follows 
\begin{align}
    \CS_{x,y}=\Tr[c_{y,x}\circ c_{x,y}]\,
\end{align}
where $x$ and $y$ are simple objects of $\CC$.

A unitary premodular category $\CC$ is an MTC if and only if its $S$-matrix is non-degenerate. Working with spin TQFTs requires introduction of super-MTCs, see \cite{Bruillard2017,Bruillard2020} for reviews. We say that a premodular category $\CC$ is a super-MTC if its M\"{u}ger center is equivalent the category of super-vector spaces $sVec$. The $S$- and $T$-matrices of $sVec$ are given by $\CS_{sVec}=\begin{psmallmatrix}
1& 1\\
1 & 1
\end{psmallmatrix}$ and $T_{sVec}=\begin{psmallmatrix}
1& 0\\
0 & -1
\end{psmallmatrix}$. In general, the super-MTC does not factorize into a product of an MTC and $sVec$. However, the $S$-matrix of any super-MTC has the form $\CS=\CS_{sVec}\otimes \CS_{nd}$ where $\CS_{nd}$ is some non-degenerate $S$-matrix \cite{Bruillard2017}.

Further, we introduce the notion of time-reversed (super-) MTC analogous to the time-reversed quantum data of Abelian Chern-Simons. The time-reversed version of a unitary (super-) MTC $\CC$, denoted by $\CC^{rev}$, also known as the reversed category, is a category with the same objects as $\CC$ but with the reversed braiding $ c_{x,y}^{rev}=c_{y,x}^{-1}$. This definition of time-reversed category was also used in \cite{Lee:2018eqa,Kong:2022cpy}.

\begin{definition}
 The unitary (super-) MTC $\CC$ is $\T$-symmetric if there exists a braided monoidal equivalence $\CF:\CC\xrightarrow{\sim}\CC^{rev}$.
\end{definition}

If such an $\CF$ exists, then $\bar \theta_x=\theta_{\CF (x)}$ and $\bar \CS_{x,y}=\CS_{\CF(x),\CF(y)}$  automatically \cite{Kong:2022cpy} for all simple objects $x$ and $y$. Here, the overline means complex conjugation and we are identifying $Hom(x,x) \cong \mathbb{C}$ for simple objects $x$. 

\paragraph{Remark.} The condition for the time-reversal invariance is somewhat similar to the conditions used for gluing left and right movers into modular invariant correlation functions \cite{Frohlich:2009gb,Moore:1988ss}.
  In RCFT, the physical Hilbert space is a finite sum of irreducible representations of chiral and anti-chiral algebras $\CH=\oplus_{i=1}\CH_{i}\otimes \CH_{\bar i}$. The partition function of the full theory is $Z=\sum_{i,j} \chi_i\, h_{i,j}\, \overline{\chi_{j}}$ where $\chi_i$ is the character of the $i$-th representation and $h_{i,j}$ are non-negative integers characterizing the field content of the theory. The set of numbers $h_{i,j}$ is subject to algebraic constraints ensuring the modular invariance \cite{Capelli1987,Dijkgraaf:1988tf,Frohlich:2009gb,Moore:1988ss}. If the theory is time-reversal invariant, then we can form a bilinear modular invariant combination of chiral characters $X=\sum_{i,j} \chi_i\, g_{i,j}\, \chi_{j}$. Note, for the general theory, there always exists the diagonal partition function $Z=\sum_{i}\chi_{i}\bar \chi_{i}$, while $X$ need not exist at all.
\subsection{(Super)-MTC for Abelian Chern-Simons} 
\label{MTCs} A prototypical example of an MTC corresponding to bosonic toral CS is the twisted group category $\CC(\sD,q)$ introduced in \cite{JOYAL199320}: objects are vector spaces labeled by the elements of the finite Abelian group $\sD$ with the corresponding fusion rules, while the braiding and associators are encoded by the quadratic form $q$, see \cite{stirling2008abelian} for a pedagogical review. Alternatively, $\CC(\sD,q)$ as a fusion category, can be seen as the category of vector bundles with convolution product. The twists and the S-matrix are determined by the quadratic and bilinear forms: 
\begin{align}
    \theta_{x}=e^{2\pi i q(x)}\,,\quad  \CS_{x,y}=e^{2\pi i b(x,y)}\,,\quad x,y\in\sD\,.
\end{align}

Any finite quadratic form defines a pointed MTC $\CC(\sD,q)$ and all pointed MTCs are of this form \cite{davydov2010Witt}.
\footnote{In physical language a ``pointed'' MTC has a fusion rule algebra that is group-like, i.e., there is exactly one simple object in the fusion product of two simple objects. The simple objects can be chiral fields in RCFT or 
anyons in 3d topological phases. Such MTCs were analyzed fairly thoroughly in Appendix E of \cite{Moore:1988qv}
and  in \cite{Lee:2018eqa}. } 
According to the connection between finite  quadratic forms and braiding on $\CC(\sD,q)$ \cite{EM1950} (see also \cite{stirling2008abelian} for a review), all $\T$-invariant MTCs correspond to $\T$-invariant quadratic forms. Therefore, Theorem \ref{Th:quadTforms} gives the complete description of $\T$-invariant pointed MTCs.

In the case of spin CS, the category assigned to a bounding circle (NS sector) is a super-MTC, while the category assigned to a non-bounding circle (R sector) is a module category for the super-MTC associated with the NS sector. According to Corollary A.19 of \cite{davydov2010Witt}, any pointed super-MTC $\CC$ is a Deligne tensor product of some pointed MTC $\CD$ and $sVec$: $\CC\xrightarrow{\sim} \CD \boxtimes sVec$.
Therefore, any pointed super-MTC is defined by a finite quadratic form since $\CD\xrightarrow{\sim} \CC(\sD,q)$ for some $(\sD,q)$. In other words, the NS sector is still defined by a quadratic form, while the R sector is defined by a quadratic refinement. 

In the spin case, the quadratic refinement $\hat q$ defines the topological spin of anyons in the R sector. The fact that the R sector is a module for the NS sector is reflected in the equivalence relation imposed on the set of quadratic refinements defined in \ref{def:QuRefsEq}. Physically, the set of values of $\hat q$ defines the set of Hall conductivities $\sigma_{H}=8\hat q$ mod $8$ \cite{Belov:2005ze}. 

Let us discuss $\T$-symmetric pointed super-MTCs. If $\CC(\sD,q)$ is $\T$-symmetric, then $\CC(\sD,q)\boxtimes sVec$ is $\T$-symmetric straightforwardly. Next, assume that $\CC(\sD,q)$ is not $\T$-invariant, but $\CC(\sD,q)\boxtimes sVec$ is. Let us observe that $sVec$ is equivalent to ``$\CC(\IZ/2,x^2/2)$''. Then, $\CC(\sD\oplus \IZ/2,q+x^2/2)$ is $\T$-invariant only if $\sD$ has a non-zero $2$-subgroup $\sD_2$: only in this case the automorphism group gets non-trivially enhanced. Further we notice that $\tau_n(q+ x^2/2)=0$ for all odd $n$, while  $\tau_{n}(q+x^2/2)=\tau_n(q)$ for all even $n$. In other words, adding $x^2/2$ can only fix a possible non-reality of $\tau_n(q)$ with $n$ odd. Therefore, $q$ must be a sum of $\omega_2^1$'s and $\omega_2^{-1}$'s. Theorem \ref{Th:quadTforms} is not applicable to degenerate quadratic forms, yet the necessary condition still holds. We simply check that  $((\IZ/2)^{r+1}, \sum_{i=1}^{r}x_{i}^2/4+y^2/2)$ is $\T$-invariant for all $r\geqslant 1$. Remarkably, this conclusion is consistent with our analysis of $\T$-invariant quadratic refinements. The R-sector is non-trivially $\T$-invariant, i.e., the $\T$-transformation involves a shift of the argument, only when the fusion group is $(\IZ/2)^r$.

Let us summarize the classification of pointed $\T$-invariant MTCs and super-MTCs. All pointed $\T$-invariant MTCs are given by $\CC(\sD,q)$ where $(\sD,q)$ is some quadratic $\T$-form. All $\T$-invariant super-MTCs are given by $\CC(\sD,q)\boxtimes \CC(\omega_2^{+1})^{\boxtimes r}\boxtimes sVec$ where $(\sD,q)$ is a quadratic $\T$-form and $r\geqslant 0$. 

\subsection{(Super)-MTC for non-Abelian Chern-Simons}
\label{Super-MTC in Abelian CS}
Consider a compact non-Abelian Lie group $G$. Classically, the Chern-Simons functional with the gauge group $G$ is defined through an element $k\in H^{4}(BG,\IZ)$, see \cite{Dijkgraaf1990}. We can associate the quantum data to such theory through the CS/WZW correspondence \cite{Axelrod1991, Elitzur:1989nr, Witten1989, Moore:1989vd}. Namely, for the bosonic CS theory with the group $G$ and positive level $k$, we associate the category $\CC(\mathfrak{g}_k)$ of the integrable highest weight representations of the corresponding Affine Lie algebra at level $k$. The category $\CC(\mathfrak{g}_k)$ is an MTC if $\mathfrak{g}$ is semi-simple. However, this is not always the case of physical interest.  Along with the categories $\CC(\mathfrak{g}_k)$, we shall use the modifications thereof obtained by gauging certain automorphisms.

\paragraph{Example.} Our basic example is $\CC(\mathfrak{su}(N)_k)$: the integrable irreps of $\mathfrak{su}(N)_k$ are labeled by Young diagrams with not more than $N\!-\!1$ rows and not more than $k$ columns. There is an action of the center $\IZ/N$ of $SU(N)$ such that the irreps labeled by the Young diagrams with $r$ boxes have the $\IZ/N$-charge $r$. By modding out the action of the center, we are left with the irreps whose Young diagrams contain $r\equiv 0$ mod $N$ boxes. The category obtained in this way is a full subcategory of $\CC(\mathfrak{su}(N)_k)$, denoted by $\CC(PSU(N)_k)$. The latter category is, in fact, a super-MTC for $N$ even: there is a ``transparent object'' $f$ with $\theta_f=-1$ labeled by the Young diagram with the maximal allowed number of boxes.
\vspace{3mm}

\paragraph{A Summary Of Known $\T$-invariant Chern-Simons theories.} Let us collect examples of $\T$-invariant non-Abelian CS theories. As gauge theories, (spin) CS theories with $G=PSU(N)$ at level $N$ were shown to be $\T$-invariant for all $N$ \cite{Hsin:2016blu} (originally this was discovered as a special case of the level-rank duality in \cite{Naculich1990,Konho1996}).  Additional examples of $\T$-invariant (spin) CS theories include the ones with the gauge groups $U(N)_{N,2N}$, $USp(2N)_N$, and $SO(N)_N$, see \cite{Hsin:2016blu,Aharony:2016jvv}. In \cite{Edie-Michell:2020}, it was demonstrated that $CS_{(SU(N)_2)}$ where $N$ is an odd number such that $\left(\frac{-1}{N}\right)=1$ is $\T$-invariant. Finally, the $\T$-invariance of $CS_{(PSU(2)_6)}$ was discussed in \cite{Fidkowski2013}. 

\subsection{Higher Gauss sums} 
Analogously to the Abelian case, we can define the higher Gauss sums for any unitary (super-) MTC $\CC$. Let $d_x=\Tr[id_x]$ be the quantum dimension of $x\in \mbox{Irr}(\CC)$ where $id_x:x\to x$ is the identity morphism of $x$ and $\mbox{Irr}(\CC)$ is the set of irreducible objects of $\CC$. For any $n\in \IZ$, we define the $n$-th Gauss sum of $\CC$ as
\begin{align}
\label{tau}
 \tau_n(\CC)=\sum_{x\in \scalebox{0.7}{Irr}(\CC)}d_x^2\; \theta_x^n\,.
\end{align}
Assuming $\tau_n(\CC)\neq 0$, the $n$-th multiplicative central charge of $\CC$ \cite{NG2019} is 
\begin{align}
\xi_n(\CC)=\frac{\tau_n(\CC)}{|\tau_n(\CC)|}\,.
\end{align}

Let us review the key properties of the higher Gauss sums and central charges. The first central charge $\xi_1(\CC)$ of any MTC is a root of unity \cite{Anderson:1988}, while the first Gauss sum of a super-MTC is always zero \cite{Bruillard2017}. In general, if $\CC$ is a only a premodular category, i.e., a category with all the precursors of the MTC except for the non-degeneracy condition, then the following relations hold \cite{NG2019}:

1) If $\tau_n(\CC)\neq 0$, then $\xi_n(\CC^{rev})=\xi_{-n}(\CC)=\overline{ \xi_n(\CC)}$,

2) For all $n\in \IZ$, $\tau_n(\CC\boxtimes \CD)=\tau_n(\CC)\tau_n(\CD)$,

3) If $\tau_n(\CC)\tau_n(\CD)\neq 0$, then $\xi_n(\CC\boxtimes \CD)=\xi_n(\CC)\xi_n(\CD)$,

4) If $\CC$ is an MTC and $\CD(\CC)$ is its Drinfeld center (see Appendix \ref{WittAppendix} for definitions), then $\xi_n(\CZ(\CC))=1$ for all $n\in \IZ$ such that $\tau_n(\CC)\neq 0$.

The higher central charges, whenever defined, are multiplicative maps which take the unit value on the Drinfeld centers. We shall see that the higher central charges can be seen as homomorphisms from the Witt group $\CW_{un}$ of pseudo-unitary non-degenerate braided fusion categories to the group of roots of unity. For future convenience, let us introduce a notation for the subset of natural numbers coprime to the Frobenius-Schur exponent \cite{Ng2007} $FSExp(\CC)$ of $\CC$: 
 \begin{align}
      D(\CC)\equiv\{l\in \IN\,|\; gcd(n,FSExp(\CC))=1 \}\,.
\end{align}
 Note, if $\CC$ is an MTC, then the Frobenius-Schur exponent of $\CC$ coincides with the order of $\theta$, the minimal natural number $m$ such that $\theta_x^m=1$ for all $x\in\mbox{Irr}(\CC)$.
 
According to Theorem 4.1 of \cite{NG2019}, if $n\in D(\CC)$, then $\tau_n(\CC)\neq 0$ and $\xi_n(\CC)$ is a root of unity. Therefore, for any unitary MTC $\CC$, we define a formal sequence of roots of unity $\{\xi_n(\CC)\}_{n\in D(\CC)}$. We recommend the original papers \cite{NG2019,ng2020higher} for a more elegant treatment of those sequences in terms of the Galois group of $\CC$. If the sequence contains only $1$'s, then such an MTC represents the trivial element of the Witt group of non-degenerate pseudo-unitary braided fusion categories $\CW_{un}$ (Theorem 4.4 of \cite{NG2019}). Thus, the sequence $\{\xi_n(\CC)\}_{n\in D(\CC)}$ is the same for all categories Witt-equivalent to $\CC$. Therefore, the $\{\xi_n(\CC)\}_{n\in D(\CC)}$ is a homomorphism from the Witt-class of $\CC$, denoted by $[\CC]$, to the group of roots of unity.

Finally, let us apply the higher Gauss sums to the study of time-reversal in unitary (super-) MTCs. An argument similar to the Abelian case, shows that each higher Gauss sum of the $\T$-invariant (super-) MTC is a real number. Since each $d_x$ is a real number for any simple object in a spherical category \cite{Etingof:2005}, the braiding reversal complex conjugates the Gauss sum $\tau_n(\CC)$.  As an example, consider the $\T$-invariant super-MTC $\CC(PSU(N)_N)$: some higher central charges for these categories can be found in Table \ref{tab:Charges}.
 \begin{table}[h!!]
\begin{center}
    \begin{tabular}{ c | c | c | c | c | c | c | c }
\backslashbox{$N$}{$m$}  & $1$ & $2$ & $3$ & $4$ & $5$ & $6$ & $7$  \\
\hline
$3$  & $-1$ & $+1$ & $-1$ & $+1$ & $-1$ & $+1$ & $-1$ \\
\hline
$5$  & $-1$ & $+1$ & $-1$ & $+1$ & $+1$ & $+1$ & $-1$ \\
\hline
$7$  & $+1$ & $+1$ & $+1$ & $+1$ & $+1$ & $+1$ & $-1$ \\
\end{tabular}
\end{center}
\caption{\label{tab:Charges}Values of $\xi_m(\CC(PSU(N)_N))$.}
\end{table}

From the preceding discussion, we obtain that the reality of all higher Gauss sums is a necessary condition for the unitary (super-) MTC $\CC$ to be $\T$-invariant. Moreover, a necessary condition for the $\T$-invariance of the MTC $\CC$ is that $[\CC]^2=1$ in $\CW_{un}$. As we discussed in the pointed case, the Witt-triviality alone is not sufficient for an MTC to be $\T$-invariant. Indeed, the Witt-class of $\CC$ is constrained by $\tau_n (\CC)$ with $n\in D(\CC)$, while $\tau_n(\CC)$ with $n\notin D(\CC)$ can be arbitrary. In the case of pointed (super-) MTCs, the reality of all the higher Gauss sums is necessary and sufficient for the time-reversal invariance. One might expect that in the non-pointed case, the reality of all higher Gauss sums would suffice as well, so we formulate 

\begin{conjecture}
\label{conjecture}
 $\CC(G_k)$ is $\T$-invariant if and only if $\tau_n(\CC(G_k))$ is a real number for all $n\in \IZ$.
\end{conjecture}
In principle, it should be possible to analyze simple MTCs directly. The proof of the conjecture, if it exists, might follow the lines of \cite{Ostrik2014} and of \cite{Edie-Michell:2020}. 

Being order two in $\CW_{un}$ is necessary for a theory to be $\T$-invariant.  A subset of $\T$-invariant theories can be obtained from Corollary 5.20 of \cite{davydov2010Witt} stating that if $\CC$ is simple non-pointed and completely anisotropic, then $[\CC]^2=1$ in $\CW$ implies that $\CC\xrightarrow{\sim}\CC^{rev}$. In other words, the smallest representative of a Witt class of order two is $\T$-invariant. 

\section{Summary and Discussion}
\label{Discussion}
 Let us briefly summarize our main results (not all of them are independent). 
 \begin{itemize}
  \item We describe isomorphism classes $\T \mbox{bil}(\IZ)$ of torsion bilinear forms on finite Abelian groups that satisfy $(\sD,b)\cong (\sD,-b)$ in Proposition \ref{Tlinking}.
  
  \item Any element of order $2$ in the Witt group $W$ of finite bilinear forms lifts to an element satisfying $(\sD,b)\cong (\sD,-b)$, according to Theorem \ref{WittTforms}. 
  
\item We characterize isomorphism classes $\T \mbox{quad}(\IZ)$ of torsion bilinear forms on finite Abelian groups that satisfy $(\sD,q)\cong (\sD,-q)$ in Theorem \ref{Th:quadTforms}.

\item We describe equivalence classes of quadratic $\T$-refinements satisfying in Section \ref{Trefinements}. Note the equivalence relations for quadratic refinements is a more relaxed notion than the isomorphism of quadratic forms, see Section \ref{Conventions}.

\item We formulate a condition for the classical Abelian spin CS theory $CS_{(\l,K)}$ to define a quantum $\T$-invariant theory in Proposition \ref{NikulinEmbedding}.

\item We classify all pointed $\T$-invariant MTCs and super-MTCs in Section \ref{MTCs}.

\item We conjecture that non-Abelian (spin) Chern-Simons with compact gauge group $G$ defined through $k\in H^4(BG,\IZ)$ is $\T$-invariant if and only if the corresponding (super-) MTC $\CC(G_k)$ has real higher Gauss sums.
\end{itemize}

\paragraph{Anti-unitary symmetries in TQFT.} In the present paper, we approach the question of when Chern-Simons admits an anti-unitary symmetry. Our approach is rather algebraic -- we look for the condition on the quantum data to admit a time-reversal symmetry. A different approach to the definition of time-reversal invariant $3\!-\!2$ TQFT should reflect the fact that such TQFT is not sensitive to changing the orientation of bordisms and $2$-manifolds. Therefore, to ensure $\T$-invariance, there must exist a lift of the corresponding tangential structure to a  non-oriented one. For simplicity, consider a non-anomalous oriented bosonic TQFT $\mathcal{F}^{SO}$, i.e., a monoidal functor from the category of oriented bordisms $\mbox{Bord}_{3,2}^{SO}$ to the category of vector spaces
\begin{align}
    \mathcal{F}^{SO}: \mbox{Bord}_{3,2}^{SO}\longrightarrow \mbox{Vec}\,.
\end{align}
Time-reversal invariance can be seen as the existence of a lift of $\mathcal{F}$ to $\tilde{\mathcal{F}}:\mbox{Bord}_{3,2}^{O}\to \mbox{Vec}$ such that the following diagram with the forgetful functor $\kappa: \mbox{Bord}_{3,2}^{SO}\to \mbox{Bord}_{3,2}^{O} $ commutes (in the 2-categorical sense)
 \begin{center}
\begin{tikzcd}
  \mbox{Bord}_{3,2}^{SO} \arrow{dr}{\mathcal{F}}  \arrow{rr}{\kappa} &                         & \mbox{Bord}_{3,2}^{O} \arrow[dl, " \tilde{\mathcal{F}}"']\\
  &  \mbox{Vec}  &
\end{tikzcd}
\end{center}
 Similarly, a non-anomalous $3\!-\!2$ spin TQFT is a functor to the category of super vector spaces
 \begin{align}
    \mathscr{S}: \mbox{Bord}_{3,2}^{Spin}\longrightarrow \mbox{sVec}\,.
\end{align}
 and $\T$-invariance condition is the existence of a lift of this functor.  Depending on the realization of the time-reversal transformation, the $Spin$ structure lifts to $Pin^{+}$ or $Pin^{-}$ structure, see \cite{Debray2022} for analysis of this question for the case of Chern-Simons with a finite Abelian gauge group.
 
 In terms of topological invariants, it was pointed out in \cite{NG2019} that the Witten-Reshetikhin–Turaev invariants of the lens space $L(n,1)$ and that of its orientation reversal are related as follows: 
 \begin{align}
     \frac{WRT_{\CC}(L(n,1))}{WRT_{\CC}(-L(n,1))}=\frac{\tau_n(\CC)}{\tau_{-n}(\CC)}\frac{\tau_{-1}(\CC)}{\tau_1(\CC)}
 \end{align}
 assuming that the relevant Gauss sums do not vanish. This relation is consistent with our conjecture, as we expect that time-reversal invariant MTCs provide us with invariants non-sensible to the orientation.
 
 In the present paper, we analyse the algebraic data coming from the value $\mathcal{F}(S^1)$. It would be interesting to investigate how the time-reversal symmetry is reflected on the value $\mathcal{F}(pt)$, in the sense of \cite{Freed:2009qp,Henriques}.
 
 \paragraph{Time-reversal symmetry vs. time-reversal duality}
\label{TinvCS}
In \cite{Freed2021}, it was shown that Witten-Reshetikhin-Turaev theory admits a non-trivial boundary theory iff it is a Turaev-Viro theory. Recall, a 3d TQFT always admits a gapless boundary e.g., through the CS/WZW correspondence, while the existence of gapped boundary theory is a non-trivial condition. For Turaev-Viro theory, the value of the TQFT functor on a 3-framed circle is a Drinfeld center. Therefore, with TQFTs admitting gapped boundaries, we associate Witt-trivial theories. This statement was rederived for Abelian Chern-Simons in \cite{kaidi2021higher}-- triviality of a certain subset of higher central charges suffices for bosonic Abelian Chern-Simons to admit a gapped boundary.

Being Witt-trivial is necessary and sufficient for the theory  to admit a gapped boundary. Let us discuss what it means for the theory to represent an order two element in the Witt group. If $[\CC(G_k)]^2=1$, then there exists a braided equivalence 
\begin{align}
\label{double}
    \CC(G_k)\boxtimes \CC(G_k)\xrightarrow{\sim} \CZ(\CA)
\end{align}
for some fusion category $\CA$ (Corollary 5.9 of \cite{davydov2010Witt}). Physically, this means that the theory $\CC(G_k)$ stacked with itself is equivalent to a Witt-trivial theory, a theory admitting a gapped boundary. By unfolding this construction, we get that $\CC(G_k)$ satisfying \eqref{double} admits a topological interface with its time-reversal: one can say that such a theory admits a time-reversal \underline{duality}. Topological interfaces in Abelian Chern-Simons were discussed in \cite{Kapustin:2010if} and studied in more detail in the context of general 3d TQFTs in \cite{Fuchs:2013}. There is a nice mathematical interpretation of such an interface pointed out to us by Christoph Schweigert: One can define a  higher version of the Morita 2-category of (algebras, bimodules, intertwiners) called  BrTens. It is a 4-category which is suitable for the Morita theory of braided fusion categories, and is spelled out in detail in   \cite{Brochier2018}. A topological interface between $\CC$ and $\CC^{rev}$ can be seen as a 1-morphism between the corresponding objects in the higher Morita category BrTens of braided tensor categories. 

Let us comment more on the analogy between pairing left- and right-movers in RCFT and the existence of time-reversal duality in CS (or self-duality in RCFT according to \cite{Moore:1989vd}). The equivalence $\CC(G_k)\boxtimes \CC(G_k)^{rev}\xrightarrow{\sim} \CZ(\CC(G_k))$, holding for any MTC $\CC(G_k)$, is related to the existence of at least one sesqui-linear modular invariant pairing on $\mbox{Irr}(\CC(G_k))$, namely the diagonal modular invariant. In its turn, condition \eqref{double} is equivalent to the existence of a bilinear pairing on objects such that the total theory is ``condensable to the vacuum'' in the sense of \cite{Kong2014AnyonCA}. 

\paragraph{$\T$-symmetry and Witt relations.}
In Section \ref{WittBilgroup}, we discussed the Witt group of bilinear forms and its relation to the Kawauchi and Kojima invariants. The latter invariants are given by the higher Gauss sums of a specific quadratic form associated to the given bilinear form. We noticed that the class of a bilinear form in $W$ depends on all the higher Gauss sums of the associated quadratic form (though, not all of them are independent). This must be contrasted with the fact that the class of an abstract quadratic form in $\CW_{pt}$ depends only on a subset of its higher Gauss sums (related to the Frobenius-Schur exponent). One might ask if there is an analog of $W$ for pseudo-unitary braided fusion categories that  depends on \underline{all} of the higher Gauss sums of a class representative. We proved that being order two in $W$ is necessary and sufficient for a bilinear form to be $\T$-symmetric in Theorem \ref{WittBilgroup}. If the conjecture made in Section \ref{conjecture} is correct, then one might ask if being order two in the hypothetical analog of $W$ is sufficient for $\T$-invariance. 

We noted that a quadratic $\T$-form can be Witt equivalent to a non-$\T$-invariant quadratic form. Similarly, the $\T$-symmetry of a unitary MTC is not preserved by the Witt relation given in Definition \ref{WittCat}. It is an intriguing question if one can define an analog of $\CW$ which would consist of $\T$-symmetric uMTCs modulo a Witt-like relation preserving the $\T$-symmetry.

\paragraph{``Real'' fusion categories.}
According to the standard definition, a fusion category is, in particular, a tensor category over $\IC$. A complex fusion category $\CC_{\IC}$ is ``real'' if it is equivalent to a product of a tensor category $\CC_{\IR}$ defined over $\IR$  and $Vec_{\IC}$: $\CC_{\IC}\xrightarrow{\sim} \CC_{\IR}\boxtimes_{\IR} Vec_{\IC}$. This is the standard framework for the theory of Galois descent, which is discussed in the categorical context in \cite{Etingof2012} and \cite{Sanford2022}. When working over $\IR$, we allow simple objects to have non-trivial endomorphism algebras. Almost trivial examples are $\CC(\phi_2)$ and $\CC(\psi_2)$: their real versions are fusion categories with two objects whose endomorphism algebras are simply $\IR$. In general, $\CC_{\IR}$ and $\CC_{\IC}$ need not have the same sets of simple objects. It would be interesting to understand if the anomaly partition function \cite{Barkeshli:2016mew,Lee:2018eqa} can detect the splitting of $\CC_{\IC}$.

\section*{Acknowledgements}
This project began with a discussion at the IAS with N. Seiberg and E. Witten in 2016. We thank A. Apte, T. D. Brennan, D. Freed, A. Schopieray, C. Schweigert, N. Seiberg, Shu-Heng Shao, C. Teleman, and Y. Tachikawa for discussions and correspondence. R.G. is thankful to A. Debray, D. Reutter, S. Sanford, L.	Stehouwer, and Y. Hu for many enlightening discussions. We are grateful to T. Johnson-Freyd for pointing a mistake in the first version of this paper. The authors are supported by the DOE under grant DOE-SC0010008 to Rutgers. R.G. is grateful to Perimeter Institute for hospitality during the school-conference Global Categorical Symmetries. Research at Perimeter Institute is supported in part by the Government of Canada through the Department of Innovation, Science and Economic Development Canada and by the Province of Ontario through the Ministry of Colleges and Universities. This work was in part performed at the Aspen Center for Physics, which is supported by National Science Foundation grant PHY-1607611.
\appendix
\section{The Maps \texorpdfstring{$\rho$}{phi}}
\label{RhoMaps}
In this appendix, we review the maps $\rho_k$ producing certain homogeneous groups out of the general Abelian group, closely following the book \cite{Miranda2009}. For a finite Abelian group $G$, we introduce a subgroup $G_{p,k}=\{x\in G\,|\,p^{k}x=0\}$ consisting of elements of order $p^k$. Then, we define
\begin{align}
    \rho_{p,k}(G)=G_{p,k}/(G_{p,k-1}+pG_{p,k})\,.
\end{align}
This map is additive and $\rho_{p,k}(G)$ is either trivial group or homogeneous group with exponent $p$: 
\begin{align}
\rho_{p,k}((\IZ/p^{r})^{d})\cong \begin{cases}
\{0\}\quad\quad\;\, \mbox{if}\;k\neq r\,,\\
(\IZ/p)^{d}\quad \mbox{if}\; k=r\,.
\end{cases}
\end{align}
 The finite form $b$ on on the $p$-group $G_p$ induces a form on $\rho_{p,k}(G_p)$, which we denote by $\rho_{p,k}(b)$, by the following rule:
\begin{align}
    \rho_{p,k}(b)(\bar{x},\bar{y})=2^{k-1}b(x,y)\,,\;\; \mbox{where}\;\; x,y\in G_p\;\;\mbox{and}\;\;\bar{x}=\rho_{p,k}(x),\,\bar{y}=\rho_{p,k}(y)\,.
\end{align}
 For example, $\rho_{2,r}(A_{2^r})=\rho_{2,r}(B_{2^r})=\rho_{2,r}(C_{2^r})=\rho_{2,r}(D_{2^r})=A_2$ and $\rho_{2,k}(E_{2^r})=\rho_{2,k}(F_{2^k})=E_2$, while for odd $p$ we have $\rho_{p,r}(X_{p^r})=X_p$ and $\rho_{p,r}(Y_{p^r})=Y_p$. 

 \section{The Seifert Invariants}
 \label{SeifertInv}
The invariants introduced by Seifiert in \cite{Seifert1933} distinguish Types I and II of normal forms on $p$-groups for odd $p$. Since $\rho_{p,k}(G)$ is a $p$-group of exponent $1$, we can treat it as a vector space over $\IZ/p$ and $(\rho_{p,k}(G),\rho_{p,k}(b))$ is an inner product space over $\IZ/p$. We embed $\IZ/p$ into $\IQ/\IZ$ by sending $1$ to $1/p$ and denote the induced form by $\tilde \rho_{p,k}(G)$. To this end, the Seifert invariant is the Legendre symbol of the determinant
 \begin{align}
     \chi_p^{k}(b)=\left(\frac{\det \tilde \rho_{p,k}(b)}{p}\right)\,.
 \end{align}
For odd $p$, the isomorphism class of a linking $((\IZ/p^r)^d),b)$ is fully determined by the invariants $r,d$, and $\chi_{p}^r(b)$. The isomorphism classes with $\chi_p^k(b)=1$ and $\chi_p^{k}(b)=-1$ can be represented by a finite form of Type II \ref{EvenType} and  Type I \ref{OddType} respectively.
\section{The Signature Invariants}
\label{SignatureInv}

Consider a $2$-group $G_2$. Due to
the relation on  the generators $A_2+E_2=3A_2$ of forms on $(\IZ/2)^d$, any form on the $2$-group of exponent $1$ can be brought to one of the two types: $dA_2$, or $(d/2)E_2$ (this is allowed only if $d$ is even). These two types are called type $A$ and type $E$, respectively.

Assume $\rho_{2,k}(b)$ is of type $E$, then we introduce a quadratic form on the group $G_2^{(k)}=G_2/\rho_{2,k}(G_2)$. Such quadratic form exists \cite{Kawauchi1980}, and defined by
\begin{align}
    q_k: G_2^{(k)}\to \IQ/\IZ\,\\
    q_k(\bar{x})=2^{k-1}b(x,x)
\end{align}
where $x$ is a pre-image of $\bar{x}$. Then, we define the Gaussian sum associated to the quadratic form:
\begin{equation}
S_k(b)=\frac{1}{\sqrt{|G_2^{(k)}|}}\sum_{\bar{x}\in G_2^{(k)}}\exp(2\pi i q_k(\bar{x}))\,.
\end{equation}
It was shown by Kawauchi and Kojima that $S_k(b)$ is the eighth root of unity $\exp(2\pi i \tilde \sigma_k(b)/8)$ whenever $\rho_{2,k}(b)$ is of type $E$.

Then, we define the signature invariant valued in the semigroup $\overline{\IZ}_8$:
\begin{align}
  \sigma_k(b)=\begin{cases}
  \infty \;\;\;\;\;\;\,\mbox{if}\;\;\mbox{$\rho_{2,k}(b)$ is of type $A$}\,,\\
  \tilde \sigma_k(b)\;\;\mbox{if}\;\;\mbox{$\rho_{2,k}(b)$ is of type $E$ or the zero-form}\,.
  \end{cases} 
\end{align}
Finally, Kawauchi and Kojima calculate the signatures for the generators of $\mbox{bil}(\IZ)_2$:
\begin{table}[h!]
\begin{center}
    \begin{tabular}{  | c | c | c | c | }
    \hline
\multirow{2}{*}{Generator}  & \multicolumn{3}{|c|}{$\sigma_k$} \\ \cline{2-4}
&$k=r$ & $r -k$ even,\;$k<r$ & $r-k$ odd,\;$k<r$\\ \hline
$A_{2^r}$ & $\infty$ & 1 & 1  \\
$B_{2^r}$ & $\infty$ & 7 &   7   \\ 
 $C_{2^r}$ & $\infty$  & 5 & 1 \\
 $D_{2^r}$ & $\infty$  & 3 & 7 \\
  $E_{2^r}$ & $0$  & 0 & 0 \\
 $F_{2^r}$ & $0$  & 0 & 4 \\
    \hline
    \end{tabular}
\end{center}
\caption{\label{GeneratorsSignatures}The signature invariants of generators of $\mbox{bil}(\IZ)_2$. We assume that $r-k\geqslant 1$, otherwise the signature is not defined.}
\end{table}

The Gauss sums are given by the standard formulas, see e.g., \cite{berndt1998gauss}:
\begin{align}
    \frac{1}{2^{r/2}}\sum_{n=0}^{2^r-1}e^{2\pi i an^2/2^{r+1}}=\exp \left(\frac{2\pi i}{8}\left(\frac{(a^2-1)(r+2)}{2}+a\right)\right)\,,\quad r\geqslant 2\,
\end{align}
for any odd $a$.

\section{The Witt Group of Non-Degenerate Braided Fusion Categories}
\label{WittAppendix}

In this appendix, we collect some relevant definitions and give a very brief review of the Witt group of braided fusion categories.  We assume the reader is familiar with the notion of a fusion category, see \cite{Etingof:2005} and \cite{EGNO} for the general reference and the original paper by A. Davydov et al. \cite{davydov2010Witt} on the Witt groups.

We say that a fusion category is pointed if all its simple objects are invertible, i.e., the fusion of simple objects is group-like. A fusion category $\CC$ is simple if the only proper fusion subcategory of $\CC$ is Vec. For a braided fusion category $\CC$, we introduce the M{\"u}ger center $\CZ_2(\CC)$ as the full subcategory whose objects have trivial double braiding with all objects in $\CC$: 
\begin{align}
    \CZ_2(\CC)=\{x\in \CC\;|\;c_{y,x}\circ c_{x,y}=id_{x\otimes y},\;\mbox{for all}\;y\in \CC\}\,.
\end{align}
We say that braided fusion category $\CC$ is non-degenerate if $\CZ_{2}(\CC)\cong \mbox{Vec}$, and slightly-degenerate if $\CZ_2(\CC)\cong \mbox{sVec}$. We define a premodular category to be a braided fusion category equipped with a ribbon structure. The ribbon and braided structures are sets of isomorphisms \eqref{braiding} and $\theta_x:x\xrightarrow{\sim} x$ subject to \eqref{balancing}.

There are some useful types of categories with less structure than premodular categories.  
A \emph{pivotal category} is a monoidal category with duals equipped with a morphism $\epsilon(x):\mathbf{1}\mapsto x\otimes x^{\vee}$ for each object $x$ with $x^{\vee}$ being the dual of $x$. The set of morphisms is subject to the conditions listed in \cite{Barrett:1993zf}. Finally, right and left traces are maps defined for any object $x$ and morphism $f:x\mapsto x$ as $\Tr_{R}[f]\equiv \epsilon(x^{\vee})\circ(\mathbf{1}\otimes f)\circ (\epsilon(x^{\vee}))^{\vee}$ and $\Tr_{L}[f]\equiv \epsilon(x)\circ(f\otimes \mathbf{1})\circ (\epsilon(x))^{\vee}$ respectively. Note these are in ${\rm End}(\textbf{1}) \cong \IC$ and hence can be identified with complex numbers. A \emph{spherical category} is a pivotal category where left and right traces are the same    for all objects \cite{Barrett:1993zf}.
Returning to premodular categories,
an alternative definition to the one given above, one could define a 
premodular category as a spherical braided fusion category. 

In this paper, we are interested in unitary categories coming from CS theories. Any unitary category is, in particular, pseudo-unitary. We refer the reader to, e.g., \cite{davydov2010Witt} for the definition of pseudo-unitary categories. It is expected that the notions of unitarity and pseudo-unitarity are equivalent\cite{Kong:2022cpy}. We use the term pseudo-unitary to be consistent with the math literature keeping in mind the conjectured equivalence. 

One of the results of \cite{Etingof:2005} states that any pseudo-unitary fusion category admits a unique spherical structure. Next, for any braided fusion category, there is a one-to-one correspondence between the spherical and ribbon structures \cite{Yetter:1992}. Therefore, unitary braided fusion categories, the categories of our interest, are premodular in a canonical way.

For any premodular category, we define the modular S-matrix:
\begin{align}
    \CS_{x,y}=\Tr[c_{y,x}\circ c_{x,y}]\,
\end{align}
where $x$ and $y$ are simple objects of $\CC$ and $\Tr$ is the trace corresponding to the ribbon structure. The non-degeneracy of the S-matrix is equivalent to the non-degeneracy of the braiding.

Let $\CC$ be a fusion category, then the Drinfeld center $\CZ(\CC)$ of $\CC$ is a  category whose objects are pairs $(x,\Phi_x)$ where $x\in Obj(\CC)$ and $\Phi_x$ is a family of natural isomorphisms 
\begin{align}
    \Phi_x: x\otimes y \xrightarrow{\sim} y\otimes x\quad \mbox{for all}\quad y\in \CC\,
\end{align}
satisfying $\Phi_{x\otimes y}=(id \otimes \Phi_y)\circ(\Phi_x\otimes id)$, see \cite{Kong:2022cpy} for the full set of axioms. The Drinfeld center of any fusion category is automatically a braided category. Moreover, the Drinfeld center of a spherical fusion category is an MTC \cite{Mueger2003}.  For any non-degenerate braided fusion category $\CC$, there is a braided monoidal equivalence $\CC\boxtimes \CC^{rev}\xrightarrow{\sim} \CZ(\CC)$ \cite{Mueger2003}. 

Let $\CC$ and $\CC'$ be fusion categories, then their Deligne tensor product $\CC\boxtimes \CC'$ is a fusion category satisfying the universal property: for any $\IC$-bilinear functor $F: \CC\times \CC'\to \CB$ for some fusion category $\CB$, there exists a unique $\IC$-linear functor $\tilde F:\CC\boxtimes \CC'\to \CB$ such that the following diagram commutes.
\begin{center}
\begin{tikzcd}
  \CC\times \CC' \arrow{dr}{F}  \arrow{rr}{\boxtimes} &                         & \CC\boxtimes \CC' \arrow[dl, "\tilde F"']\\
  &  \CB  &
\end{tikzcd}
\end{center}
Notice the universal property of the Deligne tensor product is a close analog of the universal property defining, e.g., the tensor product of vector spaces. If the categories are braided, so is their Deligne product. If $c_{x,y}$ and $c'_{x',y'}$ are the braiding isomorphisms of $\CC$ and $\CD$ respectively, then $c_{x,y}\boxtimes c'_{x,y}$ is the braiding of the tensor product. Correspondingly, if $\CS_{x,y}$ and $\CS'_{x',y'}$ are the  S-matrices of $\CC$ and $\CC'$, then the $S$-matrix of the Deligne product is simply $\tilde \CS_{x\boxtimes x',y\boxtimes y'}=\CS_{x,y}\otimes\CS'_{x',y'}$. If $\theta_x$ and $\theta'_{x'}$ are the ribbon structures on the categories $\CC$ and $\CC'$, then $\tilde\theta_{x\boxtimes x'}=\theta_x\boxtimes \theta'_{x'}$ is the ribbon structure on the product.

\begin{definition}
\label{WittCat}
 Non-degenerate braided fusion categories $\CC$ and $\CC'$ are Witt-equivalent if there exists a braided monoidal equivalence $\CC\boxtimes \CZ(\CA_1)\xrightarrow{\sim} \CC'\boxtimes \CZ(\CA_2)$ for some fusion categories $\CA_1$ and $\CA_2$. The equivalence class of $\CC$ is denoted by $[\CC]$.
\end{definition} 

Lemma $5.3$ of \cite{davydov2010Witt} states that the monoid of the Witt-equivalence classes of non-degenerate braided fusion categories equipped with the inverses $[\CC]^{-1}=[\CC^{rev}]$ is a group.
\begin{definition}
\label{WittGroups}
$\CW$ is the group of Witt-equivalence classes of non-degenerate braided fusion categories.
 $\CW_{un}$ is the subgroup of $\CW$ consisting of the Witt-equivalence classes of pseudo-unitary non-degenerate braided fusion categories. 
\end{definition} 
 Definition \ref{WittGroups} uses pseudo-unitary categories instead of unitary ones, yet it is believed \cite{Kong:2022cpy} that they are equivalent. It is also expected \cite{davydov2010Witt} that the classes of $\CW_{un}$ are well represented by the categories $\CC(\mathfrak{g}_k)$ of the integrable highest weight representations of some Affine Lie algebra $\mathfrak{g}$ at some positive level $k$. These are precisely our categories of interest. 

In Section \ref{TinvCS}, we mentioned that for special categories, $[\CC]^2=1$ in $\CW_{un}$ implies that $\CC$ is $\T$-invariant \cite{davydov2010Witt}. For the reader's reference, we collect the corresponding definitions here. A fusion category $\CC$ is completely anisotropic if the only connected \`etale algebra $A$ in $\CC$ is trivial. Recall, a unital algebra object in a fusion category is an object equipped with a pair of morphisms: the unit map and  multiplication map. An algebra object is called separable if the multiplication, as a map of $A-A$ bimodules, splits. An algebra $A\in \CC$ is called \`etale if it is commutative and separable. Finally, the algebra $A\in \CC$ is connected if $\mbox{dim}_{\IC}\mbox{Hom}_{\CC}(1,A)=1$. 
\bibliographystyle{amsalpha}

\end{document}